\documentclass[pra,aps,twocolumn,twoside,superscriptaddress,nofootinbib]{revtex4-1}
 
\usepackage[utf8]{inputenc}
\usepackage[bookmarks]{hyperref}
\hypersetup{colorlinks=true,citecolor=blue,linkcolor=blue,filecolor=blue,urlcolor=blue}
\usepackage{amsmath,amssymb,amsthm}
\usepackage{dsfont}
\usepackage{amsfonts}%
\usepackage{enumerate}
\usepackage{mathtools}
\usepackage{autonum}

\allowdisplaybreaks[4]

\newtheorem{theorem}{Theorem}[section]
\newtheorem{lemma}[theorem]{Lemma}
\newtheorem{proposition}[theorem]{Proposition} 
\newtheorem{corollary}[theorem]{Corollary}
\theoremstyle{definition}
\newtheorem{definition}[theorem]{Definition}

\newcommand{\cA}{\mathcal{A}}
\newcommand{\cB}{\mathcal{B}}
\newcommand{\cC}{\mathcal{C}}
\newcommand{\cD}{\mathcal{D}}
\newcommand{\cE}{\mathcal{E}}

\newcommand{\cH}{\mathcal{H}}

\newcommand{\cJ}{\mathcal{J}}

\newcommand{\cL}{\mathcal{L}}
\newcommand{\cM}{\mathcal{M}}
\newcommand{\cN}{\mathcal{N}}
\newcommand{\cO}{\mathcal{O}}
\newcommand{\cP}{\mathcal{P}}

\newcommand{\cS}{\mathcal{S}}
\newcommand{\cT}{\mathcal{T}}
\newcommand{\cU}{\mathcal{U}}
\newcommand{\cV}{\mathcal{V}}

\newcommand{\cZ}{\mathcal{Z}}

\newcommand{\one}{\mathds{1}}
\newcommand{\eps}{\varepsilon}

\DeclareMathOperator{\tr}{Tr}

\DeclareMathOperator{\id}{id}

\newcommand{\sumi}{\sum\nolimits}

\newcommand{\ox}{\otimes}

\DeclareMathOperator{\cov}{cov}
\DeclareMathOperator{\Had}{Had}
\newcommand{\Haddeg}{\Had_{\mathrm{deg}}}

\newcommand{\vphi}{\varphi}
\newcommand{\kH}{\mathfrak{H}}
\newcommand{\kE}{\mathfrak{E}}
\DeclareMathOperator{\eb}{EB}

\begin{document}
	\title{Approaches for approximate additivity \\ of the Holevo information of quantum channels}
	\author{Felix Leditzky}\email{felix.leditzky@jila.colorado.edu}
	\affiliation{JILA, University of Colorado/NIST, Boulder, Colorado 80309, USA}
	\affiliation{Center for Theory of Quantum Matter, University of Colorado, Boulder, Colorado 80309, USA}
	\author{Eneet Kaur}\email{ekaur1@lsu.edu}
	\affiliation{Hearne Institute for Theoretical Physics, Department of Physics and Astronomy, Baton Rouge, Louisiana 70803, USA}
	\author{Nilanjana Datta}\email{n.datta@damtp.cam.ac.uk}
	\affiliation{Department of Applied Math and Theoretical Physics, Centre for Mathematical Sciences, University of Cambridge, Cambridge CB3 0WA, UK}
	\author{Mark M.~Wilde}\email{mwilde@lsu.edu}
	\affiliation{Hearne Institute for Theoretical Physics, Department of Physics and Astronomy, Baton Rouge, Louisiana 70803, USA}
	\affiliation{Center for Computation and Technology, Louisiana State University, Baton Rouge, Louisiana 70803, USA}

\begin{abstract}
We study quantum channels that are close to another channel with weakly additive Holevo information and derive upper bounds on their classical capacity.
Examples of channels with weakly additive Holevo information are entanglement-breaking channels, unital qubit channels, and Hadamard channels.
Related to the method of approximate degradability, we define approximation parameters for each class above that measure how close an arbitrary channel is to satisfying the respective property.
This gives us upper bounds on the classical capacity in terms of functions of the approximation parameters, as well as an outer bound on the dynamic capacity region of a quantum channel.
Since these parameters are defined in terms of the diamond distance, the upper bounds can be computed efficiently using semidefinite programming (SDP).
We exhibit the usefulness of our method with two example channels: a convex mixture of amplitude damping and depolarizing noise, and a composition of amplitude damping and dephasing noise.
For both channels, our bounds perform well in certain regimes of the noise parameters in comparison to a recently derived SDP upper bound on the classical capacity.
Along the way, we define the notion of a generalized channel divergence (which includes the diamond distance as an example), and we prove that for jointly covariant channels these quantities are maximized by purifications of a state invariant under the covariance group.
This latter result may be of independent interest.
\end{abstract}

\maketitle

\section{Introduction}\label{sec:introduction}
In information theory, an imperfect communication link between a sender and a receiver is modeled as a noisy channel.
The \emph{capacity} of such a channel is defined as the largest rate at which information can be sent through the channel reliably. 
In his 1948 paper that founded information theory, \textcite{Sha48} gave a simple formula for the capacity of a channel in terms of the mutual information between an input random variable and the corresponding channel output random variable, maximized over all possible input distributions.
One of the most remarkable aspects of Shannon's formula is its \emph{single-letter} nature, meaning that the capacity of a channel only depends on the output statistics of a single use of it.

Quantum information theory generalizes the classical theory, incorporating quantum phenomena like entanglement that have the potential of enhancing communication capabilities.
Communication is modeled by quantum channels.
Depending on the type of information to be sent through a quantum channel, and the resources available to the sender and the receiver, there are various capacities characterizing a channel's capabilities.
In this paper, we focus on the task of unassisted classical information transmission through a quantum channel.
One may define a quantum analogue of the channel mutual information mentioned above, the \emph{Holevo information} (see \eqref{eq:holevo-information} in Section~\ref {sec:holevo-information-classical-capacity}), which again quantifies the maximal possible mutual information between a (classical) input random variable to a quantum channel and its (quantum) output state.
Restricting the encoding in the information-processing task so that entangled inputs across different channel uses are \emph{not} allowed, the Holevo information of a quantum channel is indeed equal to the \emph{product-state classical capacity} of the channel \cite{Hol98,SW97}.

In quantum information theory, it is natural to allow for general encodings that entangle inputs to the channel across different channel uses, since this potentially improves communication rates.
In this scenario, the situation is more complicated, as the classical capacity of a quantum channel is now given by the \emph{regularized} Holevo information \cite{Hol98,SW97}.
The regularization here means that the Holevo information needs to be evaluated for an arbitrary (unbounded) number $n$ of channel uses; the largest such value normalized by $n$ is then equal to the classical capacity.
The evaluation of the Holevo information for an unbounded number of uses of the channel renders the classical capacity intractable to compute, unless the Holevo information is additive.
In this latter case, the normalized Holevo information is the same for all $n$.

Additivity of the Holevo information is known to hold for certain classes of channels such as entanglement-breaking channels \cite{Sho02b}, unital qubit channels \cite{Kin02}, depolarizing channels \cite{Kin03}, and Hadamard channels \cite{Kin06,KMNR07}.
Moreover, a covariant qubit channel whose symmetry group forms a one-design on the output space also has additive Holevo information (see Section~\ref {sec:covariance} for details).
However, \textcite{Has09} found an example of a channel that violates additivity of the Holevo information. 
Thus, the regularization in the classical capacity formula seems to be necessary in general.
A well-known example of a qubit channel with unknown classical capacity is the amplitude damping channel. 
Deriving a tractable expression for its classical capacity remains a major open problem in quantum information theory.

A quantum channel can also be used to transmit other types of information besides classical information.
For example, the \emph{quantum capacity} of a quantum channel quantifies the highest rate at which quantum information can be sent through the channel reliably.
Similar to the classical capacity, the quantum capacity was shown to be equal to the regularization of a quantity called the coherent information \cite{Sch96,SN96,BNS98,BKN00,Llo97,Sho02,Dev05}.
However, the coherent information can be superadditive \cite{DSS98} (see also \cite{CEMOPS15} for an extreme form of this).
Thus, the regularization in the quantum capacity formula is generally necessary, rendering the capacity intractable to compute in most cases.

A similar situation holds when a quantum channel is used to transmit \emph{private} classical information.
The corresponding \emph{private capacity} can again be expressed as the regularization of a quantity called the private information \cite{Dev05,CWY04}, and the latter was found to be superadditive as well \cite{SRS08}.
However, for the class of \emph{degradable} quantum channels \cite{DS05} both the coherent information and the private information are additive (\cite{DS05} resp.~\cite{Smi08}), and in fact equal to each other \cite{Smi08}.
Subsequently, for these channels both the quantum and private capacity are equal to the coherent information.

Due to the non-additivity of the Holevo information, the coherent information, and the private information, the corresponding capacities (classical, quantum, private) are poorly understood except in a few particular cases.
Thus, to characterize the communication capabilities of most quantum channels we are left with finding good (lower and) upper bounds on their various capacities.
A powerful method to find such upper bounds on the quantum and private capacity of a channel was recently developed by \textcite{SSWR15}.
For an arbitrary channel, they defined an \emph{approximate degradability} parameter, which measures how close a channel is to being degradable (in which case it would have additive coherent and private information).
Approximately degradable channels have approximately additive coherent and private information, and this fact can be used to obtain strong upper bounds on the quantum and private capacity, respectively.
Moreover, these bounds can be easily evaluated as the degradability parameter is the solution to a semidefinite program (SDP).
The method has also been applied recently to bound various capacities of bosonic thermal channels \cite{SWAT17}.

\subsection{Main results and organization of the paper}
In the current work, we apply techniques similar to those developed in \cite{SSWR15} to the task of classical information transmission, in order to obtain upper bounds on the classical capacity of a channel.
The main mathematical tool is a continuity result for the classical capacity of quantum channels recently proved in \cite{Shi15}.
This result can be used to obtain upper bounds on the classical capacity of a channel that is close (in diamond distance) to a quantum channel with weakly additive Holevo information (Corollary~\ref {thm:approximate-additivity}).
We then define notions of approximate covariance, approximate entanglement-breaking, and approximate Hadamard-ness, that measure how far a given arbitrary channel is from satisfying the respective defining property.
In each case, Corollary~\ref {thm:approximate-additivity} gives upper bounds on the classical capacity of the channel (Corollaries~\ref {cor:cov-upper-bounds},  \ref {cor:eb-upper-bound}, and \ref {cor:hadamard-upper-bound}, respectively).
In addition, we define an alternative notion of approximate Hadamard-ness in the spirit of \cite{SSWR15}, and derive an analogous upper bound on the classical capacity in Theorem~\ref{thm:eps-hadamard}. 
This alternative approximate Hadamard-ness parameter can also be used to obtain an outer bound on the ``dynamic capcaity region'' of a channel (Theorem~\ref{thm:triple}), characterizing the ability of a quantum channel to simultaneously transmit classical and quantum information as well as generate entanglement between sender and receiver.
In all cases above, the corresponding approximation parameters can be computed (or bounded) via SDPs, and thus evaluated efficiently.

We demonstrate the usefulness of our approach on two examples of channels for which additivity of the Holevo information is not known to hold: a convex mixture of an amplitude damping channel and a depolarizing channel, and the composition of an amplitude damping channel and a $Z$-dephasing channel.
For both channels, we compare our upper bounds to an SDP upper bound on the classical capacity recently derived by \textcite{WXD16}.
Along the way, we define the notion of a generalized channel divergence that includes the diamond distance between quantum channels.
We prove in Proposition~\ref {prop:covariance} that for jointly covariant channels (i.e., two channels that are covariant with respect to the same group) any generalized channel divergence is maximized by purifications of a covariant state, i.e., a state invariant under the covariance group.
This result may be of independent interest and has been employed recently in \cite{SWAT17}.
Specializing it to the diamond distance, we obtain an analytical formula for the covariance parameter of the amplitude damping channel (Proposition~\ref {prop:amp-damp-cov-param}).

The rest of the paper is organized as follows:
In Section~\ref {sec:preliminaries} we first fix some notation, and then define and discuss the following central objects of our paper: quantum channels, the diamond norm, generalized channel divergences, the Holevo information of a quantum channel, and the classical capacity of a quantum channel.
In Section~\ref {sec:approximate-channels}, we introduce four different notions of approximate additivity of the Holevo information of a quantum channel, based on how close the channel is to being covariant, entanglement-breaking, or Hadamard, respectively.
We show how these notions lead to upper bounds on the classical capacity of a quantum channel, as well as an outer bound on the dynamic capacity region of a quantum channel.
In Section~\ref {sec:applications} we apply our results to two examples of channels and furthermore discuss why our methods do not give useful bounds for the amplitude damping channel.
Finally, we give some concluding remarks in Section~\ref {sec:conclusion}.
Appendices~\ref {sec:bitwirl}, \ref{sec:eps-hadamard} and \ref {sec:ampdamp-properties} contain the proofs of some technical results.

\section{Preliminaries}\label{sec:preliminaries}
In this paper, we use the following notation:
For a finite dimensional Hilbert space $\cH$, we denote by $\cB(\cH)$ the algebra of linear operators acting on $\cH$.
For a Hilbert space $\cH_A$ associated to a quantum system $A$, we set $|A|\coloneqq \dim\cH_A$.
We write $X_{A_1\dots A_n}$ for operators in $\cB(\cH_{A_1\dots A_n})$, where $\cH_{A_1\dots A_n} \coloneqq \cH_{A_1}\ox\cdots \ox \cH_{A_n}$.
A (quantum) state $\rho_A$ is a positive semidefinite, normalized operator, i.e., $\rho_{A}\geq 0$ and $\tr\rho_A = 1$.
A \emph{pure} state $\psi_A$ is a state with rank 1, to which we can associate a normalized vector $|\psi\rangle_A\in\cH_A$ (i.e., $\langle\psi|\psi\rangle_A=1$) such that $\psi_A = |\psi\rangle\langle \psi|_A$.
We denote by $\cP(\cH_A)\coloneqq \lbrace \rho\in\cB(\cH_A)\colon \rho\geq 0\rbrace$ the set of positive semidefinite operators on $\cH_A$, and by $\cD(\cH_A)\coloneqq \lbrace \rho\in\cP(\cH_A)\colon \tr(\rho_A)=1\rbrace$ the set of states on $\cH_A$.
For a state $\rho_A$, the \emph{von Neumann entropy} $S(\rho_A)\equiv S(A)_\rho$ is defined by $S(A)_\rho = -\tr\rho_A\log\rho_A$, where $\log$ is taken to base 2.

\subsection{Quantum channels}
A \emph{quantum channel} $\cN\colon A\to B$ is a linear, completely positive, and trace-preserving map from $\cB(\cH_A)$ to $\cB(\cH_B)$, where $\cH_A$ and $\cH_B$ are Hilbert spaces associated to the quantum systems $A$ and $B$.
We also use the notation $\cN_{A\to B}$.
We denote the identity channel on $\cB(\cH_A)$ by $\id_A$.
If a quantum channel acts on one system of a bipartite operator, we occasionally omit the identity channel, i.e., we write $\cN_{A\to B}(\rho_{RA})\equiv (\id_R\ox \cN_{A\to B})(\rho_{RA})$.
For every quantum channel $\cN\colon A\to B$ we can choose an auxiliary Hilbert space $\cH_E$, the \emph{environment}, and an isometry $V\colon \cH_A\to \cH_B\ox\cH_E$, the \emph{Stinespring dilation}, such that $\cN(\rho_A) = \tr_E (V\rho_A V^\dagger)$ \cite{Sti55}.
The isometry $V$ is unique up to left multiplication by a unitary operator acting on $\cH_E$.
A \emph{complementary channel} $\cN_c\colon A\to E$ of $\cN$ is defined by $\cN_c(\rho_A)\coloneqq \tr_B(V\rho_A V^\dagger)$ and unique up to a unitary operator acting on the output.

Let $\lbrace |i\rangle_A\rbrace_{i=1}^{|A|}$ be a basis for $\cH_A$, and define the unnormalized maximally entangled vector 
\begin{align}
|\gamma\rangle_{AA'} \coloneqq \sum_{i=1}^{|A|} |i\rangle_{A} \ox |i\rangle_{A'},\label{eq:gamma}
\end{align} 
where $A'\cong A$ (by which we mean $\cH_{A'}\cong \cH_{A}$ as Hilbert spaces).
The \emph{Choi operator} $N_{AB}$ of a quantum channel $\cN\colon A\cong A'\to B$ is defined as 
\begin{align}
N_{AB} \coloneqq (\id_{A}\ox \cN)(\gamma_{AA'}),
\end{align}
where $\id_A$ denotes the identity map on $\cB(\cH_A)$.
The Choi operator satisfies $N_{AB} \geq 0$ and $\tr_BN_{AB} = \one_{A}$, where $\one_{A}$ denotes the identity operator on $\cH_A$.
Conversely, any bipartite operator $M_{AB}$ satisfying $M_{AB}\geq 0$ and $\tr_B M_{AB}=\one_{A}$ is the Choi operator of some quantum channel $\cM\colon A\to B$.
We use this correspondence extensively throughout the paper.

Let $G$ be a group with unitary representations $U_A(g)\in\cU(\cH_A)$ on $\cH_A$ and $V_B(g)\in\cU(\cH_B)$ on $\cH_B$, respectively, where $\cU(\cH)$ denotes the unitary group acting on the Hilbert space $\cH$.
We call a quantum channel $\cN\colon A\to B$ \emph{covariant with respect to $\lbrace (U_A(g),V_B(g)) \rbrace_{g\in G}$}, if
\begin{align}
V_B(g)\, \cN(\cdot) V_B(g)^\dagger = \cN ( U_A(g) \cdot U_A(g)^\dagger)
\end{align}
for all $g\in G$.
We drop direct reference to the representations $\lbrace (U_A(g),V_B(g)) \rbrace_{g\in G}$ whenever their choice is clear from the context.

\subsection{Diamond norm}
The \emph{trace norm} $\|X\|_1$ of an operator $X\in\cB(\cH)$ is defined as
\begin{align}
\|X\|_1 \coloneqq \tr\sqrt{X^\dagger X}.
\end{align}
The diamond norm $\|\Phi\|_\diamond$ of a linear map $\Phi\colon A\to B$ is defined as
\begin{align}
\|\Phi\|_\diamond \coloneqq \max_{X_{A'A}\in\cB(\cH_{A'A})} \frac{\|(\id_{A'}\ox\Phi)(X_{A'A})\|_1}{\|X_{A'A}\|_1},
\end{align}
where $A'\cong A$.
For two quantum channels $\cN,\cM\colon A\to B$, the diamond norm $\frac{1}{2}\|\cN-\cM\|_\diamond$ of half their difference is the solution to the following semidefinite program (SDP) \cite{Wat09}:
\begin{align}
\begin{aligned}
	\text{\normalfont minimize: } & \mu\\
	\text{\normalfont subject to: } & \tr_BZ_{AB} \leq \mu \one_A\\
	& Z_{AB} \geq N_{AB} - M_{AB}\\
	& Z_{AB} \geq 0,
\end{aligned}
\label{eq:diamondnorm-SDP}
\end{align}
where $N_{AB}$ and $M_{AB}$ denote the Choi operators of the quantum channels $\cN$ and $\cM$, respectively.

\subsection{Generalized channel divergences}\label{sec:generalized-channel-divergence}
In the following, we use the notation $\cD\equiv \cD(\cH)$ and $\cP\equiv\cP(\cH)$ for the sets of density matrices and positive semidefinite operators on a generic Hilbert space $\cH$, respectively.
\begin{definition}
	[Generalized divergence; \cite{SW13,WWY14}]
	A functional $\mathbf{D}\colon \cD\times
	\cP\rightarrow\mathbb{R}$ is a \textit{generalized divergence} if it
	satisfies the monotonicity (data processing) inequality%
	\begin{align}
	\mathbf{D}(\rho\| \sigma)\geq\mathbf{D}(\mathcal{N}(\rho)\Vert
	\mathcal{N}(\sigma)),
	\end{align}
	where $\cN$ is a quantum channel.
\end{definition}

Particular examples of a generalized divergence are the quantum relative entropy $D(\rho\|\sigma)\coloneqq \tr(\rho(\log\rho-\log\sigma))$ \cite{Lin74} and the trace distance $\|\rho-\sigma\|_1$.
It follows directly from monotonicity that any generalized divergence is invariant with respect to isometries, in the sense that $\mathbf{D}(\rho \Vert\sigma)=\mathbf{D}(U\rho U^{\dag}\Vert U\sigma U^{\dag})$, where $U$ is an isometry, and that it is invariant under tensoring with another quantum state $\tau$, namely $\mathbf{D}(\rho\Vert\sigma)=\mathbf{D}(\rho\otimes\tau\Vert\sigma\otimes\tau)$. 
Note that to establish isometric invariance from monotonicity, we require a channel that can reverse the action of an isometry (see, e.g., \cite[Section~4.6.3]{Wil13} for this standard construction).

We say that a generalized channel divergence satisfies the direct-sum property with respect to classical-quantum states if the following equality holds:
\begin{multline}
\mathbf{D}\!\left(  \sumi_{x}p_{X}(x)|x\rangle\langle x|_{X}\otimes\rho
^{x}\middle\Vert\sumi_{x}p_{X}(x)|x\rangle\langle x|_{X}\otimes\sigma
^{x}\right) \\  =\sumi_{x}p_{X}(x)\mathbf{D}(\rho^{x}\Vert\sigma^{x}),
\end{multline}
where $p_{X}$ is a probability distribution, $\{|x\rangle\}_{x}$ is an orthonormal basis, and $\{\rho^{x}\}_{x}$ and $\{\sigma^{x}\}_{x}$ are sets of states. 
We note that this property holds, e.g., for trace distance and quantum relative entropy.

\begin{definition}
	[Generalized channel divergence]
	Given quantum channels $\mathcal{N}_{A\rightarrow B}$ and $\mathcal{M}_{A\rightarrow B}$, we define the generalized channel divergence as
	\begin{align}
	\mathbf{D}(\mathcal{N}\Vert\mathcal{M}) \equiv\sup_{\rho_{RA}}\mathbf{D}%
	(\mathcal{N}_{A\rightarrow B}(\rho_{RA}%
	)\Vert \mathcal{M}_{A\rightarrow B}(\rho_{RA})),
	\end{align}
	where the supremum is over all mixed states $\rho_{RA}$, and the reference system $R$ is allowed to be arbitrarily large. 
	However, as a consequence of purification, data processing, and the Schmidt decomposition, it follows that
	\begin{align}
	\mathbf{D}(\mathcal{N}\Vert\mathcal{M}) =\sup_{\psi_{RA}}\mathbf{D}%
	(\mathcal{N}_{A\rightarrow B}(\psi_{RA}%
	)\Vert \mathcal{M}_{A\rightarrow B}(\psi_{RA})),
	\end{align}
	such that the supremum can be restricted to be with respect to pure states and the reference system $R$ is isomorphic to the channel input system $A$.
\end{definition}

Particular cases of the generalized channel divergence are the diamond norm of
the difference of $\mathcal{N}_{A\rightarrow B}$ and $\mathcal{M}%
_{A\rightarrow B}$ as well as the Rényi channel divergence from \cite{CMW14}.

In the following development, the notion of joint covariance plays a central role.
We say that channels $\mathcal{N}_{A\rightarrow B}$ and $\mathcal{M}%
_{A\rightarrow B}$ are jointly covariant with respect to $\left\{  \left(
U_{A}(g),V_{B}(g)\right)  \right\}  _{g\in G}$ if each of them is covariant
with respect to $\left\{  \left(  U_{A}(g),V_{B}(g)\right)  \right\}  _{g}$.
We also use the abbreviations
\begin{align}
\mathcal{U}_{A}^{g}(\rho_{A})  &  =U_{A}(g)\rho_{A}U_{A}^{\dag}(g),\\
\mathcal{V}_{B}^{g}(\sigma_{B})  &  =V_{B}(g)\sigma_{B}V_{B}^{\dag}(g).
\end{align}

We begin with the following lemma, which will be helpful in establishing
several follow-up results:

\begin{lemma}\label{lemma:cov-critical-step}
	Let $\mathcal{N}_{A\rightarrow B}$ and
	$\mathcal{M}_{A\rightarrow B}$ be quantum channels, and let $\left\{  \left(
	U_{A}(g),V_{B}(g)\right)  \right\}  _{g\in G}$ denote unitary representations
	of a group $G$. 
	Let $\rho_{A}$ be a density operator, and let $\phi_{RA}%
	^{\rho}$ be a purification of $\rho_{A}$. 
	Let $\bar{\rho}_{A}$ denote the
	group average of $\rho_{A}$ according to a probability distribution $p_{G}$, i.e.,%
	\begin{align}
	\bar{\rho}_{A}=\sum_{g}p_{G}(g)\mathcal{U}_{A}^{g}(\rho_{A}),
	\end{align}
	and let $\phi_{RA}^{\bar{\rho}}$ be a purification of $\bar{\rho}_{A}$. 
	Moreover, for $g\in G$ we use the notation $\cN^g_{A\to B}\equiv \mathcal{V}_{B}^{g\dag}\circ\mathcal{N}_{A\rightarrow B}\circ\mathcal{U}_{A}^{g}$, and similarly for $\cM^g_{A\to B}$.
	Then the following inequality holds:
	\begin{multline}
	\mathbf{D}\left(\mathcal{N}_{A\rightarrow B}\left(\phi_{RA}^{\bar{\rho}}\right) \middle\Vert
	\mathcal{M}_{A\rightarrow B}\left(\phi_{RA}^{\bar{\rho}}\right)\right)\\
	\geq\mathbf{D}\!\left(  \sumi_{g}p_{G}(g)|g\rangle\langle g|_{P}\otimes \cN^g_{A\to B}  (\phi_{RA}^{\rho})\middle\Vert \right.\\
	 \left. \sumi_{g}p_{G}(g)|g\rangle
	\langle g|_{P}\otimes \cM^g_{A\to B}  (\phi_{RA}^{\rho})\right)\!.
	\end{multline}
	If the generalized divergence has the direct-sum property with respect to classical-quantum states, then the following inequality holds:
	\begin{multline}
	\mathbf{D}(\mathcal{N}_{A\rightarrow B}(\phi_{RA}^{\bar{\rho}})\Vert
	\mathcal{M}_{A\rightarrow B}(\phi_{RA}^{\bar{\rho}}))
	\\ \geq\sum_{g}%
	p_{G}(g)\mathbf{D}\!\left( \cN^g_{A\to B} (\phi_{RA}^{\rho
	})\middle\Vert \cM^g_{A\to B} (\phi_{RA}^{\rho})\right)\!.
	\end{multline}
\end{lemma}

\begin{proof}
	Our proof is related to an approach from \cite[Proposition 2]{TWW14} as well as that given
	in \cite{Mat10}.
	Given the purification $\phi_{RA}^{\rho}$, consider the following state%
	\begin{align}
	\left\vert \psi\right\rangle _{PRA}\equiv\sum_{g}\sqrt{p_{G}(g)}|g\rangle
	_{P}\left[  I_{R}\otimes U_{A}(g)\right]  \left\vert \phi^{\rho}\right\rangle
	_{RA}\text{.}%
	\end{align}
	Observe that $\left\vert \psi\right\rangle _{PRA}$ is a purification of
	$\bar{\rho}_{A}$ with purifying systems $P$ and $R$. By the fact that all
	purifications are related by an isometry, there exists an isometric channel
	$\mathcal{W}_{R\rightarrow PR}$ such that $\mathcal{W}_{R\rightarrow PR}%
	(\phi_{RA}^{\bar{\rho}})=\psi_{PRA}$. Then the following chain of inequalities
	holds:
	\begin{align}
	&  \mathbf{D}(\mathcal{N}_{A\rightarrow B}(\phi_{RA}^{\bar{\rho}}%
	)\Vert\mathcal{M}_{A\rightarrow B}(\phi_{RA}^{\bar{\rho}}))\nonumber\\
	&  =\mathbf{D}(\mathcal{W}_{R\rightarrow PR}(\mathcal{N}_{A\rightarrow B}%
	(\phi_{RA}^{\bar{\rho}}))\Vert\mathcal{W}_{R\rightarrow PR}(\mathcal{M}%
	_{A\rightarrow B}(\phi_{RA}^{\bar{\rho}})))\\
	&  =\mathbf{D}(\mathcal{N}_{A\rightarrow B}[\mathcal{W}_{R\rightarrow PR}%
	(\phi_{RA}^{\bar{\rho}})]\Vert\mathcal{M}_{A\rightarrow B}[\mathcal{W}%
	_{R\rightarrow PR}(\phi_{RA}^{\bar{\rho}})])\\
	&  =\mathbf{D}(\mathcal{N}_{A\rightarrow B}(\psi_{PRA})\Vert\mathcal{M}%
	_{A\rightarrow B}(\psi_{PRA}))\\
	&  \geq\mathbf{D}\!\left(  \sumi_{g}p_{G}(g)|g\rangle\langle g|_{P}%
	\otimes\left(  \mathcal{N}_{A\rightarrow B}\circ\mathcal{U}_{A}^{g}\right)
	(\phi_{RA}^{\rho})\middle\Vert \right.\nonumber\\
	&\qquad\quad \left. \sumi_{g}p_{G}(g)|g\rangle\langle g|_{P}%
	\otimes\left(  \mathcal{M}_{A\rightarrow B}\circ\mathcal{U}_{A}^{g}\right)
	(\phi_{RA}^{\rho})\right) \label{eq:classical-on-P}\\
	&  =\mathbf{D}\!\left(  \sumi_{g}p_{G}(g)|g\rangle\langle g|_{P}\otimes \cN^g_{A\to B} (\phi_{RA}^{\rho})\middle\Vert \right.\nonumber\\ 
	&\qquad\quad \left. \sumi_{g}p_{G}(g)|g\rangle
	\langle g|_{P}\otimes \cM^g_{A\to B}  (\phi_{RA}^{\rho})\right)  .
	\end{align}
	The first equality follows from isometric invariance of the channel
	divergence. The second equality follows because the isometric channel
	$\mathcal{W}_{R\rightarrow PR}$ commutes with $\mathcal{N}_{A\rightarrow B}$
	and $\mathcal{M}_{A\rightarrow B}$. The third equality follows because
	$\mathcal{W}_{R\rightarrow PR}(\phi_{RA}^{\bar{\rho}})=\psi_{PRA}$. The first
	inequality follows from monotonicity of the generalized divergence
	$\mathbf{D}$ under a dephasing of the $P$ register (where the dephasing
	operation is given by $\sum_{g}|g\rangle\langle g|\cdot|g\rangle\langle g|$).
	The last equality follows from invariance of the generalized divergence under
	unitaries, with the unitary chosen to be%
	\begin{align}
	\sum_{g}|g\rangle\langle g|_{P}\otimes V_{B}^{\dag}(g).
	\end{align}
	Note that one could also implement this operation as a classically controlled
	LOCC\ operation, i.e., a von Neumann measurement $\left\{  |g\rangle\langle
	g|\right\}  $ of the register $P$ followed by a rotation $V_{B}^{\dag}(g)$ of
	the $B$ register, as discussed in \cite[Proposition 2]{TWW14}. One can do so here because both arguments to $\mathbf{D}$ in
	\eqref{eq:classical-on-P} are classical on $P$.
\end{proof}

We then have the following proposition, which allows us to restrict the form
of the input states needed to optimize the generalized channel divergence of
two jointly covariant channels:

\begin{proposition}
	\label{prop:covariance} 
	Let $\mathcal{N}_{A\rightarrow B}$ and $\mathcal{M}%
	_{A\rightarrow B}$ be quantum channels that are jointly covariant with respect to $\left\{  \left(
	U_{A}(g),V_{B}(g)\right)  \right\}  _{g\in G}$ for a group $G$ as above. Then,
	\begin{align}
	\mathbf{D}(\mathcal{N}\Vert\mathcal{M}) =\sup_{\psi_{RA}} \left\lbrace
	\mathbf{D}(\mathcal{N}_{A\rightarrow B}%
	(\psi_{RA})\Vert \mathcal{M}_{A\rightarrow B}(\psi_{RA}))\right\rbrace,
	\end{align}
	where the supremum is over all pure states $\psi_{RA}$ such that $\psi_{A}=\frac{1}{\left\vert G\right\vert }\sumi_{g\in G}%
	\mathcal{U}_{A}^{g}\left(  \psi_{A}\right)$.
	That is, it suffices to restrict the optimization to be over pure input states
	$\psi_{RA}$ such that the reduced state $\psi_{A}$ is invariant with respect
	to the symmetrizing channel $\frac{1}{\left\vert G\right\vert }\sum
	_{g}\mathcal{U}_{A}^{g}\left(  \cdot\right)  $.
\end{proposition}

\begin{proof}
	This is an immediate consequence of Lemma~\ref{lemma:cov-critical-step}, which
	follows from the assumption of joint covariance. Applying it and taking
	$p_{G}(g)=1/\left\vert G\right\vert $, we find that%
	\begin{widetext}
	\begin{align}
	\mathbf{D}(\mathcal{N}_{A\rightarrow B}(\phi_{RA}^{\bar{\rho}}%
	)\Vert\mathcal{M}_{A\rightarrow B}(\phi_{RA}^{\bar{\rho}})) &\geq\mathbf{D}\!\left(  \sumi_{g\in G}\frac{1}{\left\vert G\right\vert }%
	|g\rangle\langle g|_{P}\otimes \cN^g_{A\to B}  (\phi_{RA}^{\rho
	})\;\middle\Vert \;\sumi_{g\in G}\frac{1}{\left\vert G\right\vert }|g\rangle\langle
	g|_{P}\otimes \cM^g_{A\to B} (\phi_{RA}^{\rho})\right) \\
	&  =\mathbf{D}\!\left(  \sumi_{g\in G}\frac{1}{\left\vert G\right\vert }%
	|g\rangle\langle g|_{P}\otimes\mathcal{N}_{A\rightarrow B}(\phi_{RA}^{\rho
	})\; \middle\Vert  \;\sumi_{g\in G}\frac{1}{\left\vert G\right\vert }|g\rangle\langle
	g|_{P}\otimes\mathcal{M}_{A\rightarrow B}(\phi_{RA}^{\rho})\right) \\
	&  =\mathbf{D}\!\left(  \mathcal{N}_{A\rightarrow B}(\phi_{RA}^{\rho
	})\middle\Vert\mathcal{M}_{A\rightarrow B}(\phi_{RA}^{\rho})\right)  .
	\end{align}
	\end{widetext}
	The first equality follows from the assumption of joint covariance, which
	implies that%
	\begin{align}
	\cN^g_{A\to B} & \equiv \mathcal{V}_{B}^{g\dag}\circ\mathcal{N}_{A\rightarrow B}\circ\mathcal{U}%
	_{A}^{g}   =\mathcal{N}_{A\rightarrow B},\\
	\cM^g_{A\to B} &\equiv \mathcal{V}_{B}^{g\dag}\circ\mathcal{M}_{A\rightarrow B}\circ\mathcal{U}%
	_{A}^{g}  =\mathcal{M}_{A\rightarrow B}.
	\end{align}
	The last inequality follows because the generalized divergence $\mathbf{D}$ is
	invariant with respect to tensoring another quantum state.
\end{proof}

Applying Proposition~\ref{prop:covariance}\ to the case in which $\left\{U_{A}(g)\right\}_{g\in G}$ is a one-design leads to the following corollary:

\begin{corollary}
	Let $\mathcal{N}_{A\rightarrow B}$ and $\mathcal{M}_{A\rightarrow B}$ be quantum channels that are
	jointly covariant with respect to $\left\{  \left(  U_{A}(g),V_{B}(g)\right)
	\right\}  _{g\in G}$ for a group $G$ and where $\left\{  U_{A}(g)\right\}
	_{g\in G}$ is a one-design. Then,
	\begin{multline}
	\mathbf{D}(\mathcal{N}\Vert\mathcal{M}) \\ =\mathbf{D}((\operatorname{id}%
	_{R}\otimes\mathcal{N}_{A\rightarrow B})(\Phi_{RA})\Vert(\operatorname{id}%
	_{R}\otimes\mathcal{M}_{A\rightarrow B})(\Phi_{RA})),
	\end{multline}
	where $|\Phi\rangle_{RA}\coloneqq |A|^{-1/2} |\gamma\rangle_{RA}$ denotes the (normalized) maximally entangled state.
\end{corollary}

\subsection{Holevo information and classical capacity of quantum channels}\label{sec:holevo-information-classical-capacity}
We define the \emph{Holevo information} $\chi(\cN)$ of a quantum channel $\cN\colon A\to B$ as
\begin{align}
\chi(\cN) &\coloneqq \max_{\cE} \chi(\cN,\cE),\label{eq:holevo-information}
\end{align}
where
\begin{multline}
\chi(\cN,\cE) \coloneqq S\left(\sumi_x p_X(x) \cN(\rho^x_A) \right)\\  - \sumi_x p_X(x) S(\cN(\rho^x_A)),
\end{multline}
and the maximum in \eqref{eq:holevo-information} is over all quantum state ensembles $\cE=\lbrace p_X(x), \rho^x_A\rbrace_x$ with $\rho^x_A\in\cD(\cH_A)$.
Note that this maximum is achieved by pure state ensembles of cardinality at most $|A|^2$, i.e., $\rho^x_A = |\psi^x\rangle\langle \psi^x|_A$ for all $x=1,\dots,|A|^2$.
Defining the classical-quantum (cq) states 
\begin{align}
\rho_{XA} &= \sumi_x p_X(x) |x\rangle \langle x|_X \ox \rho^x_A\label{eq:rho-XA}\\
\sigma_{XB} &= (\id_X\ox \cN)(\rho_{XA}),\label{eq:sigma-XB}
\end{align}
the Holevo information can be expressed as
\begin{align}
\chi(\cN) = \max_{\rho_{XA}} I(X;B)_\sigma,\label{eq:holevo-capacity-cq}
\end{align}
where the maximum is over all cq states $\rho_{XA}$ of the form in \eqref{eq:rho-XA}, the state $\sigma_{XB}$ is defined as in \eqref{eq:sigma-XB}, and $I(A;B)_\theta = S(A)_\theta + S(B)_\theta - S(AB)_\theta$ is the mutual information of a bipartite state $\theta_{AB}$.
The \emph{classical capacity} $C(\cN)$ of $\cN$ is given by the regularized Holevo information \cite{Hol98,SW97}:
\begin{align}
C(\cN) = \lim_{n\to \infty} \frac{1}{n} \chi(\cN^{\ox n}).
\label{eq:classical-capacity}
\end{align}
We say that a quantum channel $\cN$ has \emph{weakly additive} Holevo information if
\begin{align}
\chi(\cN^{\ox n}) = n\chi(\cN)
\end{align}
holds for all $n\in\mathbb{N}$.
For such a channel, the limit in the classical capacity formula \eqref{eq:classical-capacity} becomes trivial, and the classical capacity is equal to the Holevo information, $C(\cN) = \chi(\cN)$.
Furthermore, we say that a quantum channel $\cN$ has \emph{strongly additive} Holevo information if
\begin{align}
\chi(\cN\ox\cM) = \chi(\cN) + \chi(\cM)
\end{align}
for \emph{any} other channel $\cM$.
It is easy to see that every channel with strongly additive Holevo information also has weakly additive Holevo information.
Subsequently, for channels with strongly additive Holevo information we also have $C(\cN) = \chi(\cN)$.
In Section~\ref {sec:approximate-channels} we discuss examples of channels with weakly or strongly additive Holevo information.

\textcite{LS08} proved a number of continuity results for quantum channel capacities with respect to the diamond distance, including a continuity bound for the classical capacity.
\textcite{Shi15} recently refined their result on the classical capacity using techniques developed by \textcite{Win15}, as well as giving an improved continuity bound for the Holevo information of two quantum channels.
To state Shirokov's results, we introduce the function $g\colon [0,1]\to\mathbb{R}$, defined as
\begin{align}
g(\eps) \coloneqq (1+\eps)\log (1+\eps) - \eps\log\eps. \label{eq:g-function}
\end{align}
We then have:
\begin{theorem}[\cite{Shi15}]~\label{prop:continuity}
	Let $\cN,\cM\colon A\to B$ be quantum channels with $\frac{1}{2}\|\cN-\cM\|_\diamond \leq \eps$ for some $\eps\in[0,1]$.
	Then,
	\begin{enumerate}[{\normalfont (i)}]
		\item\label{item:holevo-continuity} %
		$|\chi(\cN) - \chi(\cM)|\leq \eps \log |B| + g(\eps);$
		
		\item\label{item:classical-cap-continuity} %
		$|C(\cN) - C(\cM)| \leq 2\eps \log |B| + g(\eps),$
	\end{enumerate}
	where $g(\eps)$ is defined as in \eqref{eq:g-function}.
\end{theorem}

From Theorem~\ref {prop:continuity} we can easily deduce the following result, which serves as the main mathematical tool in our discussion:

\begin{corollary}\label{thm:approximate-additivity}
	Let $\cN\colon A\to B$ be an arbitrary quantum channel, and let $\cM\colon A\to B$ be a quantum channel with weakly additive Holevo information, $\chi(\cM^{\ox n}) = n\chi(\cM)$ for all $n\in \mathbb{N}$.
	If $\frac{1}{2}\|\cN-\cM\|_\diamond \leq \eps$ for some $\eps\in [0,1]$, the classical capacity of $\cN$ can be bounded as
	\begin{align}
	C(\cN) &\leq \chi(\cM) + 2\eps \log|B| + g(\eps) \label{eq:bound-holevo-M}\\
	&\leq \chi(\cN) + 3\eps \log|B| + 2g(\eps),\label{eq:bound-holevo-N}
	\end{align}
	with $g(\eps)$ as defined in \eqref{eq:g-function}.
\end{corollary}

\section{Channels with approximately additive Holevo information}\label{sec:approximate-channels}

\subsection{Approximately covariant channels}\label{sec:covariance}
In this subsection we define a notion of approximate covariance of a quantum channel, and we show how the assumptions of Corollary~\ref {thm:approximate-additivity} can be met using this concept.
First, we discuss channels that are covariant with respect to certain groups called \emph{unitary designs}, and show how the Holevo information of these channels becomes (weakly or strongly) additive.

A group $G$ is said to form a \emph{unitary $1$-design}, if there is a unitary representation $U_A(g)\in\cU(\cH_A)$ of $G$ on $\cH_A$ such that
\begin{align}
\frac{1}{|G|} \sum_{g\in G} U_A(g) \rho_A U_A(g)^\dagger = \pi_A
\end{align}
for all $\rho_A\in\cB(\cH_A)$, where $\pi_A = \frac{1}{|A|}\one_A$ denotes the completely mixed state on $\cH_A$.

A group $G$ is said to form a \emph{unitary $2$-design}, if there is a unitary representation $U_A(g)\in\cU(\cH_A)$ of $G$ on $\cH_A$ such that for all quantum channels $\Lambda\colon A\to A$,
\begin{multline}
\frac{1}{|G|} \sum_{g\in G} U_A(g) \Lambda\left(U_A(g)^\dagger\cdot U_A(g)\right) U_A(g)^\dagger \\ = \int_{\cU(\cH_A)} d\mu(U)\, U \Lambda\left(U^\dagger\cdot U\right) U^\dagger,
\end{multline}
where $d\mu(U)$ denotes the Haar measure on $\cU(\cH_A)$.
Equivalently \cite{DCEL09}, for all $\rho_{A'A}$ (with $A'\cong A$),
\begin{multline}
\frac{1}{|G|} \sum_{g\in G} (U_{A'}(g)\ox U_A(g))\rho_{A'A} (U_{A'}(g)\ox U_A(g))^\dagger \\ = \int_{\cU(\cH_A)} d\mu(U)\,(U\ox U)\rho_{A'A} (U\ox U)^\dagger.
\end{multline}

For covariant channels with a one-design as the input space representation, the formula for the Holevo information simplifies in the following way:
\begin{lemma}[\cite{Hol13}]\label{lem:1-design-holevo}
	Let $\cN\colon A\to B$ be a quantum channel that is covariant with respect to $\lbrace (U_A(g),V_B(g)) \rbrace_{g\in G}$, where the representation $U_A(g)\in\cU(\cH_A)$ is a one-design.
	Then,
	\begin{align}
	\chi(\mathcal{N})=S(\mathcal{N}(\pi))-\min_{\psi}S(\mathcal{N}(\psi)).
	\end{align}
\end{lemma}

For covariant qubit-qubit channels with a one-design as the output space representation, the Holevo information is weakly additive:
\begin{lemma}[\cite{Kin02}]\label{lem:1-design-additivity}
	Let $\cN\colon A\to B$ be a qubit-qubit channel, $|A|=|B|=2$, that is covariant with respect to $\lbrace (U_A(g),V_B(g)) \rbrace_{g\in G}$, and where $V_B(g)\in\cU(\cH_B)$ is a one-design.
	Then the Holevo information of $\cN$ is weakly additive,
	\begin{align}
	\chi(\cN^{\ox n}) = n\chi(\cN)\qquad \text{for all $n\in\mathbb{N}$}.
	\end{align}
	Consequently, $C(\cN) = \chi(\cN)$ for such covariant channels.
\end{lemma}

\begin{proof}
	First, let $A$ and $B$ be arbitrary systems with $|A|=|B|$ (i.e., not necessarily qubits), and let $\cN\colon A\to B$ be a channel that is covariant with respect to $\lbrace (U_A(g),V_B(g)) \rbrace_{g\in G}$, where the representation $V_B(g)\in\cU(\cH_B)$ is a one-design.
	Then $\cN$ is unital, $\cN(\one_A) = \one_B$:
	\begin{align}
	\cN(\one_A) &= |A|\, \cN(\pi_A)\\
	&= \frac{|A|}{|G|} \sum_{g\in G} \cN\left( U_A(g) \pi_A U_A(g)^\dagger \right)\\
	&= \frac{|A|}{|G|} \sum_{g\in G} V_B(g)\cN\left(\pi_A \right) V_B(g)^\dagger\\
	&= |A|\pi_B\\
	&= \one_B.
	\end{align}
	For qubits $A$ and $B$ (i.e., $|A|=|B|=2$), the result now follows from King's result about unital qubit channels \cite{Kin02}.
\end{proof}

For covariant quantum channels $\cN\colon A\to B$, where $A$ and $B$ are isomorphic, $d$-dimensional systems, the Holevo information is additive if the group representations are unitary two-designs.
This result is a direct consequence of \cite{DCEL09,MGE11,Kin03}, and we give a proof in Appendix~\ref {sec:bitwirl} for the sake of completeness.
\begin{lemma}\label{lem:2-design-additivity}
	Let $A$ and $B$ be isomorphic $d$-dimensional quantum systems, $A\cong B$, let $G$ be a group, and let $\cN\colon A\to B$ be a quantum channel that is covariant with respect to a unitary representation $U_A(g)\in\cU(\cH_A)$ of $G$ on $\cH_A$ resp.~$\cH_B$ that is a unitary two-design.
	Then $\cN$ is an $|A|$-dimensional depolarizing channel, and hence its Holevo information is strongly additive,
	\begin{align}
	\chi(\cN\ox \cM) = \chi(\cN) + \chi(\cM)
	\end{align}
	for an arbitrary quantum channel $\cM\colon A'\to B'$.
	Consequently, $C(\cN) = \chi(\cN)$ for such covariant channels.
\end{lemma}

We now introduce a notion of approximate covariance.
For a group $G$ with unitary representations $U_A(g)\in\cU(\cH_A)$ on $\cH_A$ and $V_B(g)\in\cU(\cH_B)$ on $\cH_B$, respectively, and an arbitrary quantum channel $\cN\colon A\to B$, the \emph{twirled channel} $\cN_G$ of $\cN$ is defined as
\begin{align}
\cN_G \coloneqq \frac{1}{|G|} \sum_{g\in G} V_B(g)^\dagger \cN( U_A(g) \cdot U_A(g)^\dagger) V_B(g).
\label{eq:twirled-channel}
\end{align}
This twirled channel $\cN_G$ is covariant with respect to $\lbrace (U_A(g),V_B(g))\rbrace_{g\in G}$ by construction.
Our notion of approximate covariance of a quantum channel is based on how close the channel is in diamond norm to its twirled channel:

\begin{definition}[Approximate covariance]\label{def:approximate-covariance}
		We fix a group $G$ with unitary representations $U_A(g)\in\cU(\cH_A)$ on $\cH_A$ and $V_B(g)\in\cU(\cH_B)$ on $\cH_B$.
		For a given $\eps\in [0,1]$, we call a channel $\cN$ \emph{$\eps$-covariant with respect to $\lbrace (U_A(g),V_B(g)) \rbrace_{g\in G}$}, if
		\begin{align}
		\frac{1}{2}\| \cN - \cN_G\|_\diamond \leq \eps.
		\end{align}
		We define the \emph{covariance parameter} $\cov_G(\cN)$ as the smallest $\eps\geq0$ such that $\cN$ is $\eps$-covariant with respect to the given representations of $G$.
\end{definition}

For given representations of a group $G$, the covariance parameter $\cov_G(\cN)$ can be efficiently computed using the SDP in \eqref{eq:diamondnorm-SDP} for the diamond norm. 
The Choi operator $N_{AB}^G$ of the twirled channel $\cN_G$ can be obtained from the Choi operator $N_{AB}$ via the relation
\begin{align}
N_{AB}^G = \frac{1}{|G|} \sum_{g\in G} (\bar{U}_A(g)\ox V_B(g)) N_{AB} (\bar{U}_A(g)\ox V_B(g))^\dagger. \label{eq:twirled-choi-state}
\end{align}

By Lemma~\ref {lem:1-design-additivity}, the Holevo information of qubit-qubit channels that are covariant with respect to one-designs is weakly additive.
More generally, quantum channels acting on a $d$-dimensional system which are covariant with respect to a two-design have strongly additive Holevo information, since these channels are $d$-dimensional depolarizing channels (Lemma~\ref {lem:2-design-additivity}).
In view of Definition~\ref {def:approximate-covariance}, we can therefore apply Corollary~\ref {thm:approximate-additivity} to any quantum channel once we consider suitable groups (or rather representations thereof) and the corresponding twirled channel.

It is well known that the ``projective'' \emph{Pauli group}\footnote{Since for example $XY = iZ$, one needs to add the phases $\pm1, \pm i$ to turn $\cP$ into a group. However, under the action of $\cP$ by conjugation these phases drop out, and it suffices to just consider $\cP$.} $\cP = \lbrace \one, X, Y, Z\rbrace$ forms a one-design,
\begin{align}
\frac{1}{4} (\rho + X\rho X + Y \rho Y + Z\rho Z) = \frac{1}{2} \one
\end{align}
for all $\rho\in \cB(\mathbb{C}^2)$, where $X,Y,Z$ are the usual Pauli matrices.
For an $n$-qubit system $A$ (i.e., $|A|=2^n$), the Pauli group $\cP_n$ consists of all $n$-fold tensor products of elements in $\cP$, i.e., $\cP_n\coloneqq \cP^{\ox n}$. 
The ``projective'' \emph{Clifford group} $\cC_n$ is defined as the normalizer of $\cP_n$ in $\cU(\cH_A)$, that is,
\begin{align}
\cC_n = \lbrace U\in \cU(\cH_A)\colon U \cP_n U^\dagger = \cP_n\rbrace.
\end{align}
The Clifford group $\cC_n$ forms a unitary two-design \cite{DLT02} (see also \cite{DCEL09}).
More generally, the Clifford group $\cC^{\text{HW}}_d$ of the Heisenberg-Weyl group %
acting on a $d$-dimensional system is also a unitary two-design \cite{CLLW16}.

Combining the above observations with Lemma~\ref {lem:1-design-additivity}, Lemma~\ref {lem:2-design-additivity}, and Corollary~\ref {thm:approximate-additivity}, we obtain the following:
\begin{corollary}\label{cor:cov-upper-bounds}
   Let $A$ and $B$ be quantum systems with $|A|=|B|$, and let $\cN\colon A\to B$ be an arbitrary quantum channel.
   \begin{enumerate}[{\normalfont (i)}]
   	\item\label{item:cov-qubits} For $|A|=|B|=2$, the choices $G=\cP$ and $\eps = \cov_\cP(\cN)$ yield
   	\begin{align}
   	C(\cN) &\leq \chi(\cN_\cP) + 2\eps + g(\eps)\\
   	&\leq \chi(\cN) + 3\eps + 2g(\eps),
   	\end{align}
   	where $g(\eps)$ is defined through \eqref{eq:g-function}.
   	  	
   	\item\label{item:n-qubits} For $|A|=|B|=2^n$, the choices $G = \cC_n$ and $\eps = \cov_{\cC_n}(\cN)$ yield
   	\begin{align}
   	C(\cN) &\leq \chi(\cN_{\cC_n}) + 2n\eps + g(\eps)\\
   	&\leq \chi(\cN) + 3n\eps + 2g(\eps).
   	\end{align}
   	   	
   	\item\label{cov:qudits} For $|A|=|B|=d$, the choices $G = \cC^{\mathrm{HW}}_d$ and $\eps = \cov_{\cC^{\mathrm{HW}}_d}(\cN)$ yield
   	\begin{align}
   	C(\cN) &\leq \chi(\cN_{\cC^{\mathrm{HW}}_d}) + 2\eps\log d + g(\eps)\\
   	&\leq \chi(\cN) + 3\eps\log d + 2 g(\eps).
   	\end{align}
   \end{enumerate}
\end{corollary}
Note that in all three cases of Corollary~\ref {cor:cov-upper-bounds}, the Holevo information $\chi(\cN_G)$ of the twirled channel $\cN_G$ for $G\in \lbrace \cP, \cC_n, \cC^{\text{HW}}_d\rbrace$ can be computed via Lemma~\ref {lem:1-design-holevo}.

\subsection{Approximately entanglement-breaking channels}\label{sec:entanglement-breaking}

A channel $\cN\colon A\to B$ is called \emph{entanglement-breaking} \cite{HSR03}, if $(\id_R\ox \cN)(\rho_{RA})$ is separable for all auxiliary quantum systems $R$ and input states $\rho_{RA}$.
Equivalently, $\cN$ is entanglement-breaking if and only if its Choi operator $N_{AB}$ is separable.
Entanglement-breaking channels are an important class of quantum channels with strongly additive Holevo information \cite{Sho02b}:
If $\cN\colon A\to B$ is entanglement-breaking, then
\begin{align}
\chi(\cN\ox \cM) = \chi(\cN) + \chi(\cM)
\label{eq:eb-additivity}
\end{align}
for any arbitrary quantum channel $\cM\colon A'\to B'$.
In particular, we have $\chi(\cN^{\ox n}) = n\chi(\cN)$ for entanglement-breaking channels, and thus $C(\cN) = \chi(\cN)$ by \eqref{eq:classical-capacity}.
Hence, an appropriate notion of an approximately entanglement-breaking channel, together with \eqref{eq:eb-additivity}, provides another way of bounding the classical capacity of an arbitrary quantum channel via Corollary~\ref {thm:approximate-additivity}.

Fix quantum systems $A$ and $B$, and let $\kE \equiv \kE(A\to B) = \lbrace \cN\colon A\to B\mid N_{AB}\text{ is separable} \rbrace$ denote the set of all entanglement-breaking channels.

\begin{definition}\label{def:approximate-entanglement-breaking}
	For a given $\eps\in [0,1]$, a quantum channel $\cN\colon A\to B$ is called \emph{$\eps$-entanglement-breaking}, if
	\begin{align}
	\min_{\cM\in\kE} \frac{1}{2}\| \cN - \cM \|_\diamond \leq \eps.
	\end{align}
	We define the \emph{entanglement-breaking parameter} $\eb(\cN)$ to be the smallest $\eps\in[0,1]$ such that $\cN$ is $\eps$-entanglement-breaking.
\end{definition}

With this definition, we have the following application of Corollary~\ref {thm:approximate-additivity}:

\begin{corollary}\label{cor:eb-upper-bound}
	Let $\cN\colon A\to B$ be a quantum channel, and set $\eps\coloneqq \eb(\cN)$.
	Then, 
	\begin{align}
	C(\cN) &\leq \chi(\cM) + 2\eps\log|B| + g(\eps)\\
	&\leq \chi(\cN) + 3\eps\log|B| + 2g(\eps),
	\end{align}
	where the function $g(\eps)$ is defined through \eqref{eq:g-function}, and $\cM\colon A\to B$ is the entanglement-breaking channel achieving the minimum in Definition~\ref {def:approximate-entanglement-breaking}.
\end{corollary}

If $|A||B|\leq 6$, a state $\rho_{AB}$ is separable if and only if it has positive partial transpose (PPT) \cite{HHH96}, i.e., $\rho_{AB}^{T_B}\geq 0$, where $T_B$ denotes the transpose on the $B$ system.
Hence, in low dimensions we can efficiently compute $\eb(\cN)$:
\begin{lemma}
	Let $\cN\colon A\to B$ be a quantum channel.
	If $|A||B|\leq 6$, then the entanglement-breaking parameter $\eb(\cN)$ is the solution to the following SDP:
	\begin{align}
	\begin{aligned}
	\text{\normalfont minimize: } & \mu\\
	\text{\normalfont subject to: } & \tr_BZ_{AB} \leq \mu \one_A\\
	& Z_{AB} \geq N_{AB} - M_{AB}\\
	& Z_{AB} \geq 0\\
	& M_{AB} \geq 0\\
	& \tr_BM_{AB} = \one_A\\
	& M_{AB}^{T_B} \geq 0.
	\end{aligned}
	\label{eq:EB-SDP}
	\end{align}
\end{lemma}

\begin{proof}
	The constraint $M_{AB}^{T_B} \geq 0$ forces the operator $M_{AB}$ to be PPT, which by the assumption $|A||B|\leq 6$ is equivalent to $M_{AB}$ being separable.
	Hence, $M_{AB}$ corresponds to the Choi operator of an entanglement-breaking channel $\cM\colon A\to B$, and the SDP in \eqref{eq:EB-SDP} computes $\eb(\cN)$.
\end{proof}

\subsection{Approximately Hadamard channels}\label{sec:hadamard}

A quantum channel $\cN\colon A \to B$ is called \emph{Hadamard} if its complementary channel $\cN_c\colon A\to E$ is entanglement-breaking.
Since a channel is entanglement-breaking if and only if its Kraus operators are rank-one operators \cite{HSR03}, we can parametrize a Hadamard channel as follows (cf.~\cite{Wil13}): 
Let $K$ be the Kraus rank of $\cN$, i.e., the minimal number of Kraus operators, which is equal to the smallest dimension of the environment $|E|$ \cite{Wol12}. 
Let $\lbrace |\vphi_k\rangle_{E}\rbrace_{k=1}^K$ be a collection of normalized vectors on $\cH_E$, $\langle \vphi_k|\vphi_k\rangle = 1$ for all $k=1,\dots,K$, let $\cS=\lbrace |\psi_k\rangle_A\rbrace_{k=1}^K$ be an overcomplete set of vectors on $\cH_A$ (such that $\sum_{k=1}^K |\psi_k\rangle\langle \psi_k|_A = \one_A$), and let $\lbrace |i\rangle_B\rbrace_{i=1}^{|B|}$ be an orthonormal basis for $\cH_B$.\footnote{For simplicity, we assume $|B|=K=|E|$ here, though this is not necessary.}
These sets of vectors give rise to a Hadamard channel $\cN\colon A\to B$, whose action on an input state $\rho_A$ is given by
\begin{align}
\cN(\rho_A) = \sum_{k,\,l} \langle \psi_k| \rho_A |\psi_l\rangle \langle \vphi_l|\vphi_k\rangle_E |k\rangle \langle l|_B.
\label{eq:hadamard}
\end{align}
Define $[\rho_A]_{\cS}$ as the matrix representation of $\rho_A$ with respect to $\cS=\lbrace |\psi_k\rangle_A\rbrace_{k=1}^K$, i.e., $([\rho_A]_{\psi})_{ij} = \langle \psi_i|\rho_A|\psi_j\rangle$.
Furthermore, define the matrix $(\Gamma)_{ij} = \langle \vphi_i|\vphi_j\rangle_E$.
Then \eqref{eq:hadamard} can be rewritten as
\begin{align}
\cN(\rho_A) = \Gamma^\dagger \ast [\rho_A]_{\cS},
\label{eq:hadamard-product}
\end{align}
where $\ast$ denotes the Hadamard product (viz.~element-wise product) of matrices: $(A\ast B)_{ij} \coloneqq (A)_{ij} (B)_{ij}$.
Equation \eqref{eq:hadamard-product} is the reason such channels are called Hadamard channels.

Hadamard channels are another class of quantum channels with strongly additive Holevo information \cite{Kin06,KMNR07}:
If $\cN$ is a Hadamard channel, then for any other channel $\cM$ we have
\begin{align}
\chi(\cN\ox \cM) = \chi(\cN) + \chi(\cM).\label{eq:hadamard-additivity}
\end{align}
Hence, an appropriate notion of approximately Hadamard channels leads to another application of Corollary~\ref {thm:approximate-additivity} for bounding the classical capacity of arbitrary quantum channels.

In what follows, we define two notions of  approximately Hadamard channels, one called $\varepsilon$-close Hadamard channels and the other called
$\varepsilon$-Hadamard channels. 
To define $\varepsilon$-close Hadamard channels,
fix quantum systems $A$ and $B$ and let $\kH$ be the set of Hadamard channels from $A$ to~$B$.

\begin{definition}\label{def:approximate-hadamard}
	For a given $\eps\in[0,1]$, a quantum channel $\cN\colon A\to B$ is called $\eps$-close-Hadamard, if
	\begin{align}
	\min_{\cM\in \kH}\frac{1}{2}\|\cN - \cM\|_\diamond \leq \eps.
	\end{align} 
	We define the \emph{Hadamard parameter} $\Had(\cN)$ to be the smallest $\eps\in[0,1]$ such that $\cN$ is $\eps$-close-Hadamard.
\end{definition}

For reasons that will become clear shortly, it is useful to define the following upper bound on $\Had(\cN)$.
Let $\cS=\lbrace |\psi_k\rangle_A\rbrace_{i}$ be an overcomplete set of vectors on $\cH_A$, i.e., $\sum_i |\psi_i\rangle\langle \psi_i|_A = \one_A$, and define the set
\begin{align}
\kH_\cS \coloneqq \lbrace \cN\in \kH\colon \cN(\rho_A) = \Gamma^\dagger \ast [\rho_A]_{\cS} \rbrace \subseteq \kH,
\end{align}
where $\Gamma\geq 0$ satisfies $\Gamma_{kk}=1$ for all $k$.
We then consider the parameter
\begin{align}
\Had_\cS(\cN) \coloneqq \min_{\cM\in\kH_\cS} \frac{1}{2} \|\cN - \cM\|_\diamond,
\label{eq:approximate-hadamard-S}
\end{align}
which clearly satisfies $\Had(\cN) \leq \Had_\cS(\cN)$.
In addition, $\Had_\cS(\cN)$ is the solution to an SDP and thus efficiently computable:
\begin{lemma}\label{lem:Had-S-SDP}
	For fixed $\cS$, the parameter $\Had_\cS(\cN)$ is the solution to the following SDP:
	\begin{align}
	\begin{aligned}
	\text{\normalfont minimize: } & \mu\\
	\text{\normalfont subject to: } & \tr_BZ_{AB} \leq \mu \one_A\\
	& Z_{AB} \geq N_{AB} - \cO^\Gamma_{AB}\\
	& Z_{AB} \geq 0\\
	& \Gamma \geq 0\\
	& \langle k|\Gamma |k\rangle = 1\text{ for all $k$},
	\end{aligned}
	\label{eq:H_S-SDP}
	\end{align}
	where $\cO^\Gamma_{AB} \coloneqq \sumi_{i,\,j,\,k,\,l}(\Gamma\ast [|i\rangle\langle j|]_\cS)_{kl} \, |i\rangle \langle j|_A \ox |k\rangle \langle l|_B$.
\end{lemma}

\begin{proof}
	A matrix $\Gamma$ is positive semidefinite if and only if it is the Gram matrix of a collection of vectors, i.e., $(\Gamma)_{ij} = \langle\vphi_i|\vphi_j\rangle$ for some set of vectors $\lbrace|\vphi_i\rangle\rbrace_i$. 
	Hence, we can optimize over positive semidefinite $\Gamma$, and the constraints $\langle k|\Gamma |k\rangle = 1$ for all $k$ enforce that $\langle \vphi_k|\vphi_k\rangle=1$.
	Moreover, since the Hadamard product $\ast$ is distributive over addition, $(A+B)\ast C = A\ast C + B\ast C$, the Choi operator $\sumi_{i,\,j,\,k,\,l}(\Gamma\ast [|i\rangle\langle j|]_\cS)_{kl} \, |i\rangle \langle j|_A \ox |k\rangle \langle l|_B$ of a Hadamard channel $\cM$ is linear in $\Gamma$ for a fixed overcomplete set of vectors $\cS = \lbrace |\psi_k\rangle_A\rbrace_{i}$.
	Thus, the SDP in \eqref{eq:H_S-SDP} optimizes the diamond distance $\frac{1}{2}\|\cN-\cM\|_\diamond$ over all $\cM\in\kH_\cS$.
\end{proof}

We can combine Definition~\ref {def:approximate-hadamard}, Corollary~\ref {thm:approximate-additivity}, and \eqref{eq:hadamard-additivity} to obtain an upper bound on $C(\cN)$ similar to Corollary~\ref {cor:cov-upper-bounds}.
Since the function $2\eps \log |B| + g(\eps)$ is monotonically increasing for $\eps\in[0,1]$, we can use the parameter $\Had_\cS(\cN)$ for a fixed $\cS$ instead of $\Had(\cN)$.
The former has the advantage of being efficiently computable via its SDP representation given in Lemma~\ref {lem:Had-S-SDP}.
In summary, we have the following:
\begin{corollary}\label{cor:hadamard-upper-bound}
	Let $\cN\colon A\to B$ be a quantum channel, and set $\eps=\Had(\cN)$ and $\eps_\cS\coloneqq \Had_\cS(\cN)$ for a fixed collection $\cS=\lbrace |\psi_k\rangle_A\rbrace_k$ of vectors.
	Then, 
	\begin{align}
	C(\cN) &\leq \chi(\cM) + 2\eps\log|B| + g(\eps)\\
	&\leq \chi(\cN) + 3\eps\log|B| + 2g(\eps)\\
	&\leq \chi(\cN) + 3\eps_\cS\log|B| + 2g(\eps_\cS),
	\end{align}
	where the function $g(\eps)$ is defined through \eqref{eq:g-function}, and $\cM\colon A\to B$ is the Hadamard channel achieving the minimum in Definition~\ref {def:approximate-hadamard}.
	Note that we also have 
	\begin{align}
	C(\cN) \leq \chi(\cM_\cS) + 2\eps_\cS\log|B| + g(\eps_\cS),
	\end{align}
	where $\cM_\cS\colon A\to B$ is the Hadamard channel defined in terms of $\cS$ achieving the minimum in \eqref{eq:approximate-hadamard-S}.
\end{corollary}

We now define $\varepsilon$-Hadamard channels, an alternative way of defining a Hadamard parameter of an arbitrary quantum channel.
To this end, we recall that Hadamard channels $\cN\colon A\to B$ belong to the class of degradable channels mentioned in Section~\ref {sec:introduction}: 
For every Hadamard channel $\cN\colon A\to B$ with environment $E$ there exists a quantum channel $\cD\colon B\to E$, called the \emph{degrading channel}, such that $\cN_c = \cD\circ \cN$.
In other words, there exists a quantum channel that Bob can apply locally to simulate the leakage of $\cN$ to the environment, modeled by the complementary channel $\cN$.
Moreover, the degrading channel of Hadamard channels is entanglement-breaking; that is, the Choi operator $D_{BE}$ of $\cD$ is separable \cite{BHTW10}.
This observation gives rise to the following definition:

\begin{definition}\label{def:approximate-hadamard-deg}
	For a given $\eps\in[0,1]$, a quantum channel $\cN\colon A\to B$ is called $\eps$-Hadamard, if
	\begin{align}
	\min_{\cD\in\kE(B\to E)} \frac{1}{2}\|\cN^c - \cD\circ \cN\|_\diamond \leq \eps,
	\end{align}
	where $\kE(B\to E)$ denotes the set of all entanglement-breaking channels from $B$ to $E$.
	The Hadamard parameter $\Haddeg(\cN)$ is defined to be the smallest $\eps\in[0,1]$ such that $\cN$ is $\eps$-Hadamard.
\end{definition}

If $|B||E|\leq 6$, the parameter $\Haddeg(\cN)$ of a quantum channel $\cN\colon A\to B$ with environment $E$ can be expressed as the solution to an SDP.
This SDP is obtained from the SDP that computes the approximate degradability parameter defined in \cite{SSWR15} by adding a constraint enforcing the `approximate degrading' map $\cD$ to be entanglement-breaking.
In summary, we have:
\begin{lemma}\label{lem:Hadamard-deg-sdp}
	Let $\cN\colon A\to B$ be a quantum channel with environment $E$.
	If $|B||E|\leq 6$, then the Hadamard parameter $\Haddeg(\cN)$ is the solution to the following SDP:
	\begin{align}
	\begin{aligned}
	\text{\normalfont minimize: } & \mu\\
	\text{\normalfont subject to: } & \tr_E Z_{AE} \leq \mu \one_A\\
	& Z_{AE} \geq N_{AE} - \cJ(\cD\circ \cN)\\
	& Z_{AE} \geq 0\\
	& D_{BE} \geq 0\\
	& \tr_E D_{BE} = \one_B\\
	& D_{BE}^{T_E} \geq 0,
	\end{aligned}
	\label{eq:Haddeg-SDP}
	\end{align}
	where $\cJ(\cdot)$ denotes the Choi-Jamio\l kowski isomorphism that maps a channel to its Choi operator.
\end{lemma}

Note that in \eqref{eq:Haddeg-SDP} the operator $D_{BE}$ corresponds to the Choi operator of $\cD\colon B\to E$, and the constraint $D_{BE}^{T_E} \geq 0$ enforces $D_{BE}$ to be PPT, which for $|B||E|\leq 6$ is equivalent to $\cD$ being entanglement-breaking.
The map $\cS\mapsto\cJ(\cS\circ \cT)$ is linear in $\cS$ for a fixed channel $\cT$ \cite{SSWR15}, and hence, the corresponding constraint in \eqref{eq:Haddeg-SDP} is indeed semidefinite.

Using the Hadamard parameter $\Haddeg(\cdot)$ from Definition~\ref {def:approximate-hadamard-deg}, we can again deduce an upper bound on the classical capacity of a quantum channel by adapting the arguments in the proof of Corollary~\ref {thm:approximate-additivity}. In particular, we have the following theorem, whose proof we give in Appendix~\ref{sec:eps-hadamard}.

\begin{theorem}\label{thm:eps-hadamard}
For a channel $\mathcal{N}\colon A\rightarrow B$, set $\varepsilon = \Haddeg(\cN)$, and let $\mathcal{D}\colon B\rightarrow E$ be the entanglement-breaking channel achieving the minimum in Definition~\ref{def:approximate-hadamard-deg}. Then%
\begin{multline}
C(\mathcal{N})\leq\max_{\{p_{X}(x),\psi^{x}\}}\left[  H(FE)_{\xi
}-H(E|X)_{\xi}\right]  \\+\left(  2\varepsilon\log|E|+g(\varepsilon
)\right)  ,
\end{multline}
where the cq state $\xi_{XFE}$ is defined as
\begin{align}
\xi_{XFE}\coloneqq \sumi_{x}p_{X}(x)|x\rangle\langle x|_{X}\otimes
(\mathcal{U}_{B\rightarrow FE}^{\mathcal{D}}\circ\mathcal{N}%
_{A\rightarrow B})(\psi_{A}^{x}),
\end{align}
$\mathcal{U}_{B\rightarrow FE}^{\mathcal{D}}$ is a Stinespring dilation of $\cD\colon B\to E$, and $\{ \psi^x_A\}_x$ is a set of pure states.
\end{theorem}

\subsection{Bounds on triple trade-off capacities for $\varepsilon$-close Hadamard
channels}

In addition to the various single-resource capacities that we have discussed
so far, one can also consider the ability of a quantum channel to generate or consume
multiple resources \cite{Shor_CE,DS05,DHW05RI,HW10,HW10b,BHTW10,WH12}. One of the most general such tasks is the triple trade-off
problem \cite{HW10,HW10b}:\ What is the net rate at which classical bits, quantum bits, and
entangled bits can be consumed, in addition to many independent uses of a
quantum channel, in order to generate the same resources? Let $C$ denote the
rate of classical communication, $Q$ the rate of quantum communication, and
$E$ the rate of entanglement generation (if the rate is negative, it means
that the resource is being consumed).
The set of all achievable rates in this setting is called the ``dynamic capacity region.'' This question has been addressed in \cite{HW10,HW10b,WH12} (see also \cite{Wil13}),
via the following dynamic quantum capacity theorem:

\begin{theorem}
[{\normalfont \cite{HW10,WH12}}]The quantum dynamic capacity region $\mathcal{C}_{\operatorname{CQE}%
}(\mathcal{N})$ of a quantum channel $\mathcal{N}$ is given by the following
expression:%
\begin{align}
\mathcal{C}_{\operatorname{CQE}}(\mathcal{N})=\overline{\bigcup_{k=1}^{\infty
}\frac{1}{k}\mathcal{C}_{\operatorname{CQE}}^{(1)}(\mathcal{N}^{\otimes k})}, \label{eq:reg-3-tradeoff}
\end{align}
where the overbar indicates the closure of a set. The region $\mathcal{C}%
_{\operatorname{CQE}}^{(1)}(\mathcal{N})$ is equal to the closure of the union
of the state-dependent\ regions $\mathcal{C}_{\operatorname{CQE}}%
^{(1)}(\mathcal{N},\sigma)$:%
\[
\mathcal{C}_{\operatorname{CQE}}^{(1)}(\mathcal{N})\coloneqq\overline{\bigcup
_{\sigma}\mathcal{C}_{\operatorname{CQE}}^{(1)}(\mathcal{N},\sigma)}.
\]
The state-dependent\ region $\mathcal{C}_{\operatorname{CQE}}%
^{(1)}(\mathcal{N},\sigma)$ is the set of all rates $C$, $Q$, and $E$, such that%
\begin{align}
C+2Q &  \leq I(AX;B)_{\sigma},\label{eq-tr:CQ-bound}\\
Q+E &  \leq I(A\rangle BX)_{\sigma},\label{eq-tr:QE-bound}\\
C+Q+E &  \leq I(X;B)_{\sigma}+I(A\rangle BX)_{\sigma}.\label{eq-tr:CQE-bound},%
\end{align}
where the coherent information $I(F\rangle G)_\omega$ of a bipartite state $\omega_{FG}$ is defined as $I(F\rangle G)_\omega \coloneqq H(G)_\omega - H(F G)_\omega$.
The above entropic quantities are with respect to a classical--quantum state
$\sigma_{XAB}$, where%
\begin{align}
\sigma_{XAB}\coloneqq\sum_{x}p_{X}(x)|x\rangle\langle x|_{X}\otimes\mathcal{N}%
_{A^{\prime}\rightarrow B}(\phi_{AA^{\prime}}^{x}%
),\label{eq-tr:main-theorem-state}%
\end{align}
and the states $\phi_{AA^{\prime}}^{x}$ are pure, with system $A$ isomorphic to system $A'$. It is implicit that one
should consider states on $A^{\prime k}$ instead of $A^{\prime}$ when taking
the regularization in \eqref{eq:reg-3-tradeoff}.
\end{theorem}

If the channel $\mathcal{N}$\ is a Hadamard channel, then the regularization
is not needed and the dynamic capacity region is equal to $\mathcal{C}%
_{\operatorname{CQE}}^{(1)}(\mathcal{N})$ \cite{BHTW10,WH12}. Thus, we are motivated to consider
the notion of approximate Hadamard channels in this context, leading to Theorem~\ref{thm:eps-hadamard} below.
Note that Gao et al.~\cite{GJL15} derived an outer bound on the dynamic capacity region of a special class of channels using operator space methods.

\begin{theorem}\label{thm:triple}
Let $\mathcal{N}$ be an $\varepsilon$-close-Hadamard channel; i.e., there
exists a Hadamard channel $\mathcal{M}$ such that%
\[
\frac{1}{2}\Vert\mathcal{N}-\mathcal{M}\Vert_{\diamond}\leq\varepsilon.
\]
Then the dynamic capacity region $\mathcal{C}_{\operatorname{CQE}}(\mathcal{N})$ is
contained in the region $\mathcal{C}_{\operatorname{CQE}}^{(1)}(\mathcal{M}%
,\varepsilon)$, where $\mathcal{C}_{\operatorname{CQE}}^{(1)}(\mathcal{M}%
,\varepsilon)$\ is equal to the union of the state-dependent\ regions
$\mathcal{C}_{\operatorname{CQE}}^{(1)}(\mathcal{M},\tau,\varepsilon)$:%
\[
\mathcal{C}_{\operatorname{CQE}}^{(1)}(\mathcal{M},\varepsilon)\coloneqq\overline
{\bigcup_{\tau}\mathcal{C}_{\operatorname{CQE}}^{(1)}(\mathcal{M},\tau%
,\varepsilon)}.
\]
The state-dependent\ region $\mathcal{C}_{\operatorname{CQE}}%
^{(1)}(\mathcal{M},\tau,\varepsilon)$ is the set of all rates $C$, $Q$, and $E$,
such that%
\begin{align}
C+2Q &  \leq I(AX;B)_{\tau}+f_{1}(\varepsilon),\\
Q+E &  \leq I(A\rangle BX)_{\tau}+f_{2}(\varepsilon),\\
C+Q+E &  \leq I(X;B)_{\tau}+I(A\rangle BX)_{\sigma}+f_{2}(\varepsilon),
\end{align}
where%
\begin{align}
f_{1}(\varepsilon)  & \coloneqq 2\varepsilon\log|B|+g(\varepsilon), \label{eq:f1-func}\\
f_{2}(\varepsilon)  & \coloneqq 2\varepsilon\log|B|+g(\varepsilon)+2\sqrt
{2\varepsilon}\log|E|+g(\sqrt{2\varepsilon}), \label{eq:f2-func}
\end{align}
and $g(\varepsilon)$ is defined in \eqref{eq:g-function}.
The above entropic quantities are with respect to a classical--quantum state
$\tau_{XAB}$, where%
\[
\tau_{XAB}\coloneqq \sum_{x}p_{X}(x)|x\rangle\langle x|_{X}\otimes\mathcal{M}%
_{A^{\prime}\rightarrow B}(\phi_{AA^{\prime}}^{x}),
\]
and the states $\phi_{AA^{\prime}}^{x}$ are pure.
\end{theorem}

\begin{proof}
To prove this result, we consider an approach related to that for a statement given in \cite{Shi15}, recalled here as
 Theorem~\ref{prop:continuity}.  It is convenient to distinguish between the output system $B$ of
the channel $\mathcal{N}$ and the output system $\tilde{B}$ of $\mathcal{M}$,
where $\tilde{B}\cong B$. For $n\in\mathbb{N}$, let%
\begin{align}
\rho_{XAA^{\prime n}}=\sum_{x}p_{X}(x)|x\rangle\langle x|_{X}\otimes
\phi_{AA^{\prime n}}^{x}%
\end{align}
be an arbitrary cq state, where each $\phi_{AA^{\prime n}}^{x}$ is a pure
state. By assumption, we have that%
\begin{align}
\frac{1}{2}\Vert\mathcal{N}_{A^{\prime}\rightarrow B}-\mathcal{M}_{A^{\prime
}\rightarrow\tilde{B}}\Vert_{\diamond}\leq\varepsilon
,\label{eq:diamond-norm-bnd-3-trade-off}%
\end{align}
and so by the continuity of Stinespring's representation theorem \cite{KSW08}, it
follows that there exist isometric channels $\mathcal{U}_{A^{\prime
}\rightarrow BE}^{\mathcal{N}}$ and $\mathcal{U}_{A^{\prime}\rightarrow
BE}^{\mathcal{M}}$ extending $\mathcal{N}_{A^{\prime}\rightarrow B}$ and
$\mathcal{M}_{A^{\prime}\rightarrow\tilde{B}}$, respectively, such that the
corresponding complementary channels $\mathcal{N}_{c}$ and $\mathcal{M}_{c}$
satisfy%
\begin{align}
\frac{1}{2}\left\Vert \mathcal{N}_{c}%
-\mathcal{M}_{c}\right\Vert _{\diamond
}\leq\sqrt{2\varepsilon}.\label{eq:diamond-norm-bnd-3-trade-off-comp}%
\end{align}
We thus define the following cq states for $1\leq t\leq n-1$:%
\[
\sigma_{XAB^{n}E^{n}}^{0}\coloneqq \left(  \operatorname{id}_{XA}\otimes\left[
\mathcal{U}_{A^{\prime}\rightarrow BE}^{\mathcal{N}}\right]  ^{\otimes
n}\right)  (\rho_{XAA^{\prime n}}),
\]
\vspace{-.3in}
\begin{multline}
\sigma_{XA\tilde{B}_{\leq t}\tilde{E}_{\leq t}B_{>t}E_{>t}}^{t}\coloneqq \\
\left(  \operatorname{id}_{XA}\otimes\left[  \mathcal{U}_{A^{\prime
}\rightarrow\tilde{B}\tilde{E}}^{\mathcal{M}}\right]  ^{\otimes t}%
\otimes\left[  \mathcal{U}_{A^{\prime}\rightarrow BE}^{\mathcal{N}}\right]
^{\otimes(n-t)}\right)  (\rho_{XAA^{\prime n}}),
\end{multline}
\vspace{-.3in}
\[
\sigma_{XA\tilde{B}^{n}\tilde{E}^{n}}^{n}\coloneqq \left(  \operatorname{id}%
_{XA}\otimes\left[  \mathcal{U}_{A^{\prime}\rightarrow\tilde{B}\tilde{E}%
}^{\mathcal{M}}\right]  ^{\otimes n}\right)  (\rho_{XAA^{\prime n}}),
\]
where we use the shorthand $Q_{\leq i}\equiv Q_{1}Q_{2}\dots Q_{i}$ for $1\leq
i$ and quantum systems $Q_{1},\dots,Q_{i}$, and use analogous definitions for
$Q_{<i}$, $Q_{\geq i}$, and $Q_{>i}$.

For a Hadamard channel $\mathcal{M}$ and for all integers $n>1$, there exist
states $\omega_{A_{i}X_{i}\tilde{B}_{i}\tilde{E}_{i}}^{i}$ for $1\leq i\leq
n$, each having the form%
\[
\omega_{A_{i}X_{i}\tilde{B}_{i}\tilde{E}_{i}}^{i}\coloneqq \sum_{x_{i}}p_{X_{i}}%
(x_{i})|x_{i}\rangle\langle x_{i}|_{X_{i}}\otimes\mathcal{U}_{A_{i}^{\prime
}\rightarrow\tilde{B}_{i}\tilde{E}_{i}}^{\mathcal{M}}(\phi_{A_{i}A_{i}%
^{\prime}}^{x_{i}}),
\]
such that the following entropy inequalities hold \cite{WH12,Wil13}%
\begin{align}
I(AX;\tilde{B}^{n})_{\sigma^{n}}  & \leq\sum_{i=1}^{n}I(A_{i}X_{i};\tilde
{B}_{i})_{\omega^{i}}=nI(\tilde{A}\tilde{X};\tilde{B}|Z)_{\omega},\\
H(\tilde{B}^{n}|X)_{\sigma^{n}}  & \leq\sum_{i=1}^{n}H(\tilde{B}_{i}%
|X_{i})_{\omega^{i}}=nH(\tilde{B}|\tilde{X}Z)_{\omega},\\
H(\tilde{B}^{n})_{\sigma^{n}}  & \leq\sum_{i=1}^{n}H(\tilde{B}_{i}%
)_{\sigma^{n}}=nH(\tilde{B}|Z)_{\omega},\\
-H(\tilde{E}^{n}|X)_{\sigma^{n}}  & \leq-\sum_{i=1}^{n}H(\tilde{E}_{i}%
|X_{i})_{\omega^{i}}=-nH(\tilde{E}|\tilde{X}Z)_{\omega},\label{eq:hadamard-entropy-ineqs}
\end{align}
where the state $\omega_{Z\tilde{A}\tilde{X}\tilde{B}\tilde{E}}$ is defined as%
\[
\omega_{Z\tilde{A}\tilde{X}\tilde{B}\tilde{E}}\coloneqq \frac{1}{n}\sum_{i=1}%
^{n}|i\rangle\langle i|_{Z}\otimes\omega_{A_{i}X_{i}\tilde{B}_{i}\tilde{E}%
_{i}}^{i},
\]
and the registers $\tilde{A}$ and $\tilde{X}$ are taken to be large enough to
contain the contents of the largest $A_{i}$ and $X_{i}$, respectively.

Applying \eqref{eq:diamond-norm-bnd-3-trade-off} and
\eqref{eq:diamond-norm-bnd-3-trade-off-comp}, we find that
\begin{align}
\frac{1}{2}\Vert\sigma_{XA\tilde{B}_{\leq t-1}B_{>t-1}}^{t-1}-\sigma
_{XA\tilde{B}_{\leq t}B_{>t}}^{t}\Vert_{1}  & \leq\varepsilon
,\label{eq:bob-continuity-states}\\
\frac{1}{2}\Vert\sigma_{XA\tilde{E}_{\leq t-1}E_{>t-1}}^{t-1}-\sigma
_{XA\tilde{E}_{\leq t}E_{>t}}^{t}\Vert_{1}  & \leq\sqrt{2\varepsilon
},\label{eq:eve-continuity-states}%
\end{align}
for $1\leq t\leq n$ by the definition of the diamond norm. Moreover,
$\sigma_{XA\tilde{B}_{\leq t-1}B_{>t-1}}^{t-1}$ and $\sigma_{XA\tilde{B}_{\leq
t}B_{>t}}^{t}$ have the same marginals on the systems $XA\tilde{B}_{<t}B_{>t}%
$, and $\sigma_{XA\tilde{E}_{\leq t-1}E_{>t-1}}^{t-1}$ and $\sigma
_{XA\tilde{E}_{\leq t}E_{>t}}^{t}$ have the same marginals on the systems
$XA\tilde{E}_{<t}E_{>t}$. Hence, using the continuity bound from \cite[Corollary~1]{Shi15}
we obtain%
\begin{multline}
\left\vert I(AX;B_{t}|\tilde{B}_{<t}B_{>t})_{\sigma^{t-1}}-I(AX;\tilde{B}%
_{t}|\tilde{B}_{<t}B_{>t})_{\sigma^{t}}\right\vert \\
\leq f_{1}(\varepsilon) \label{eq:shi-cont}
\end{multline}
for $1\leq t\leq n$, where $\tilde{B}_{<1}$ and $B_{>n}$ represent empty
systems, respectively, and where
$f_{1}(\varepsilon)$ is defined in \eqref{eq:f1-func}. We also get that%
\begin{multline}
\left\vert H(B_{t}|\tilde{B}_{<t}B_{>t}X)_{\sigma^{t-1}}-H(\tilde{B}%
_{t}|\tilde{B}_{<t}B_{>t}X)_{\sigma^{t}}\right\vert \\
\leq2\varepsilon\log|B|+g(\varepsilon),
\end{multline}%
\begin{multline}
\left\vert H(E_{t}|\tilde{E}_{<t}E_{>t}X)_{\sigma^{t-1}}-H(\tilde{E}%
_{t}|\tilde{E}_{<t}E_{>t}X)_{\sigma^{t}}\right\vert \\
\leq2\sqrt{2\varepsilon}\log|E|+g(\sqrt{2\varepsilon}),
\end{multline}
by applying \eqref{eq:bob-continuity-states},
\eqref{eq:eve-continuity-states}, and \cite[Lemma~2]{Win15}. We can now bound $I(AX;B^{n}%
)_{\sigma^{0}}$ from above as follows:
\begin{align}
&  I(AX;B^{n})_{\sigma^{0}}\nonumber\\
&  =I(AX;\tilde{B}^{n})_{\sigma^{n}}+I(AX;B^{n})_{\sigma^{0}}-I(AX;\tilde
{B}^{n})_{\sigma^{n}}\nonumber\\
&  =I(AX;\tilde{B}^{n})_{\sigma^{n}}+\sum_{t=1}^{n}\Bigg[  I(AX;B_{t}|\tilde{B}_{<t}B_{>t})_{\sigma^{t-1}%
}\nonumber\\
&\qquad \qquad -I(AX;\tilde{B}_{t}|\tilde{B}_{<t}B_{>t})_{\sigma^{t}}\Bigg]
\label{eq:telescope}\\
&  \leq I(AX;\tilde{B}^{n})_{\sigma^{n}}+nf_{1}(\varepsilon)\label{eq:bound}\\
&  \leq\sum_{i=1}^{n}I(A_{i}X_{i};\tilde{B})_{\omega^{i}}+nf_{1}%
(\varepsilon)\nonumber\\
&  =nI(\tilde{A}\tilde{X};\tilde{B}|Z)_{\omega}+nf_{1}(\varepsilon)\nonumber\\
&  \leq nI(\tilde{A}\tilde{X}Z;\tilde{B})_{\omega}+nf_{1}(\varepsilon
).\nonumber
\end{align}
where \eqref{eq:telescope} follows from writing out the telescope sum and
observing that all terms except the first and last one cancel out. Step
\eqref{eq:bound} uses \eqref{eq:shi-cont} in each term of the
sum. The next inequality and the equality after that follow by invoking \eqref{eq:hadamard-entropy-ineqs}.
The final inequality is a consequence of the chain rule and non-negativity of mutual information.

For the second bound involving the coherent information term $I(A\rangle
B^{n}X)_{\sigma^{0}}$, consider that%
\begin{align}
I(A\rangle B^{n}X)_{\sigma^{0}}=H(B^{n}|X)_{\sigma^{0}}-H(E^{n}|X)_{\sigma
^{0}}.
\end{align}
We handle each term separately. First, 
by reasoning similar to that in the proof of \cite[Theorem~3.4]{SSWR15}, consider that%
\begin{align}
&  H(B^{n}|X)_{\sigma^{0}}\\
&  =H(\tilde{B}^{n}|X)_{\sigma^{n}}+H(B^{n}|X)_{\sigma^{0}}-H(\tilde{B}%
^{n}|X)_{\sigma^{n}}\\
&  =H(\tilde{B}^{n}|X)_{\sigma^{n}}\\
&  \quad+\sum_{t=1}^{n}\left[  H(B_{t}|\tilde{B}_{<t}B_{>t}X)_{\sigma^{t-1}%
}-H(\tilde{B}_{t}|\tilde{B}_{<t}B_{>t}X)_{\sigma^{t}}\right]  \\
&  \leq H(\tilde{B}^{n}|X)_{\sigma^{n}}+n\left(  2\varepsilon\log
|B|+g(\varepsilon)\right)  \\
&  \leq\sum_{i=1}^{n}H(\tilde{B}|X_{i})_{\omega^{i}}+n\left(  2\varepsilon
\log|B|+g(\varepsilon)\right)  \\
&  =nH(\tilde{B}|\tilde{X}Z)_{\omega}+n\left(  2\varepsilon\log
|B|+g(\varepsilon)\right)  .
\end{align}
By similar reasoning, we find that%
\begin{multline}
-H(E^{n}|X)_{\sigma^{0}}\leq-nH(\tilde{E}|\tilde{X}Z)_{\omega}\\
+n\left(  2\sqrt{2\varepsilon}\log|E|+g(\sqrt{2\varepsilon})\right)  .
\end{multline}
Thus, we conclude that%
\[
\frac{1}{n}I(A\rangle B^{n}X)_{\sigma^{0}}\leq H(\tilde{B}|\tilde{X}%
Z)_{\omega}-H(\tilde{E}|\tilde{X}Z)_{\omega}+f_{2}(\varepsilon),
\]
where $f_{2}(\varepsilon)$ is defined in 
\eqref{eq:f2-func}.
For the final term $I(X;B^{n})_{\sigma^{0}}+I(A\rangle B^{n}X)_{\sigma^{0}}$,
consider that%
\[
I(X;B^{n})_{\sigma^{0}}+I(A\rangle B^{n}X)_{\sigma^{0}}=H(B^{n})_{\sigma^{0}%
}-H(E^{n}|X)_{\sigma^{0}}.
\]
Since we already have a bound for $-H(E^{n}|X)_{\sigma^{0}}$, we just need to
bound $H(B^{n})_{\sigma^{0}}$. By applying similar reasoning as above, we find
that%
\begin{align}
H(B^{n})_{\sigma^{0}}  & \leq H(\tilde{B}^{n})_{\sigma^{n}}+n\left(
2\varepsilon\log|B|+g(\varepsilon)\right)  \\
& \leq\sum_{i=1}^{n}H(\tilde{B}_{i})_{\sigma^{n}}+n\left(  2\varepsilon
\log|B|+g(\varepsilon)\right)  \\
& =nH(\tilde{B}|Z)_{\omega}+n\left(  2\varepsilon\log|B|+g(\varepsilon
)\right)  .
\end{align}
Combining with the above, we then get that%
\begin{multline}
\frac{1}{n}\left[  H(B^{n})_{\sigma^{0}}-H(E^{n}|X)_{\sigma^{0}}\right]  \\
\leq H(\tilde{B}|Z)_{\omega}-H(\tilde{E}|\tilde{X}Z)_{\omega}+f_{2}%
(\varepsilon).
\end{multline}
Considering that%
\begin{align}
H(\tilde{B}|Z)_{\omega}-H(\tilde{E}|\tilde{X}Z)_{\omega}  & =I(\tilde
{X};\tilde{B}|Z)_{\omega}+I(A\rangle B\tilde{X}Z)_{\omega}\\
& \leq I(\tilde{X}Z;\tilde{B})_{\omega}+I(A\rangle B\tilde{X}Z)_{\omega}%
\end{align}
completes the proof.
\end{proof}

We end this section by mentioning that this kind of approach could be applied in the context of Hadamard broadcast channels,
which constitute the main class of quantum broadcast channels for which we have single-letter
capacity regions \cite{WDW17,QW17}.

\section{Application to particular qubit channels}\label{sec:applications}

The MATLAB code used to obtain the numerical results of this section is available in the ancillary files section on the arXiv page of this paper.
The code makes heavy use of the quantinf package \cite{Cub}, as well as the YALMIP package to solve SDPs \cite{Lof04}.

\subsection{On the amplitude damping channel}
The amplitude damping channel $\cA_p$ is defined for $p\in[0,1]$ as
\begin{align}
\cA_p(\rho) &= K_1\rho K_1^\dagger + K_2 \rho K_2^\dagger,
\label{eq:amplitude-damping}
\end{align}
where
\begin{align} 
K_1 &= |0\rangle \langle 0| + \sqrt{1-p} |1\rangle\langle 1|\\
K_2 &= \sqrt{p}|0\rangle \langle 1|.
\end{align}
An exact expression for the Holevo information of the amplitude damping channel was derived by \textcite{GF05}.
They showed that $\chi(\cA_p)$ is equal to
\begin{align}
\max_{q\in [0,1]} \left\lbrace h((1-p)q) - h\!\left(\frac{1+\sqrt{1-4p(1-p)q^2}}{2}\right) \right\rbrace\!,
\end{align}
where $h(x)\coloneqq -x\log x - (1-x)\log (1-x)$ denotes the binary entropy.
However, its classical capacity $C(\cA_p)$ is yet to be determined, which remains a major open problem in quantum information theory.

For low values of $p$, the amplitude damping channel belongs to the class of so-called \emph{low-noise channels} \cite{LLS17}.
This is the content of Proposition~\ref {prop:ampdamp-low-noise} below, whose proof we defer to Appendix~\ref {sec:ampdamp-properties}.
\begin{proposition}\label{prop:ampdamp-low-noise}
	For $p\in[0,1]$, we have
	\begin{align}
	\left\Vert \operatorname{id}-\mathcal{A}_{p}\right\Vert _{\diamond}=2p.
	\end{align}
\end{proposition}
We have $\cA_0=\id$, the identity channel, which is trivially covariant under the full unitary group, and furthermore Hadamard, since the complementary channel of the identity channel is completely depolarizing and hence entanglement-breaking.
It is also easy to see that $\cA_1$ is entanglement-breaking.
Therefore, one might expect to get useful upper bounds on the classical capacity of $\cA_p$ using notions of approximate covariance or Hadamard-ness for $p\gtrsim 0$, or approximate entanglement-breaking for $p\lesssim 1$.

The results of \cite{LLS17} can be applied to show that a (super)linear behavior in the underlying noise parameter of an approximation parameter leads to a useless bound on the classical capacity, as the appearing error term involves the function $\eps\log\eps$, which has infinite slope at $\eps = 0$ if $\eps=O(p)$.
Unfortunately, all parameters that we introduced in Section~\ref {sec:approximate-channels} have such behavior in the amplitude damping parameter $p$: the covariance parameter $\cov_\cP(\cA_p)$, the entanglement-breaking parameter $\eb(\cA_p)$, and the Hadamard parameters $\Had(\cA_p)$ and $\Haddeg(\cA_p)$.
First, the developments in Section~\ref {sec:generalized-channel-divergence} lead to the following analytical formula for $\cov_\cP(\cA_p)$, which we prove in Appendix~\ref {sec:ampdamp-properties}:
\begin{proposition}\label{prop:amp-damp-cov-param}
	The covariance parameter $\cov_\cP(\mathcal{A}_{p})$ of an amplitude
	damping channel $\cA_p$ with damping parameter $p\in\left[  0,1\right]  $ with respect to the Pauli group $\cP$ is given by
	\begin{align}
	\cov_\cP(\mathcal{A}_{p})=\frac{p}{2}.
	\label{eq:amp-damping-cov}%
	\end{align}
	
\end{proposition}

We can also derive an analytical lower bound on the entanglement-breaking parameter $\eb(\cA_p)$, showing that the latter has at least linear behavior in $p$.
The proof of the following result is again deferred to Appendix~\ref {sec:ampdamp-properties}:

\begin{proposition}\label{prop:amp-damp-eb-param}
	The entanglement-breaking parameter $\eb(\cA_p)$ of an amplitude damping channel $\cA_p$ with damping parameter $p\in [0,1]$ satisfies
	\begin{align}
	\eb(\cA_p) \geq \frac{(1-p)(2\sqrt{1-p}-p)}{4(1-p)-p^2} \geq \frac{1-p}{2}.
	\end{align}
\end{proposition}

Moreover, as shown in Figure~\ref {fig:ampdamp-hadamard}, numerics demonstrate that both Hadamard parameters $\Had_{\cS}(\cA_p)$ (with $\cS=\lbrace |0\rangle,|1\rangle\rbrace$) and $\Haddeg(\cA_p)$ are at least linear in $p$.

\begin{figure}[t]
	\centering
	\includegraphics{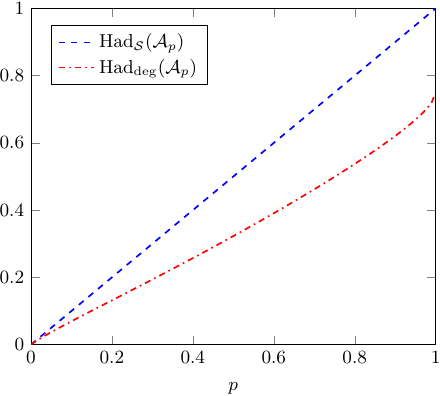}
	\caption{Plot of the Hadamard parameters of the amplitude damping channel $\cA_p$. The Hadamard parameter $\Had_{\cS}(\cA_p)$ with $\cS=\lbrace |0\rangle,|1\rangle\rbrace$, as given in \eqref{eq:approximate-hadamard-S}, is plotted in blue/dashed, and the Hadamard parameter $\Haddeg(\cA_p)$, as given in Definition~\ref {def:approximate-hadamard-deg}, is plotted in red/dash-dotted.}
	\label{fig:ampdamp-hadamard}
\end{figure}

Recently, \textcite{WXD16} derived the following strong converse upper bound on the classical capacity of a quantum channel $\cN\colon A\to B$:
\begin{align}
C(\cN) \leq C_\beta(\cN) \coloneqq \log \beta(\cN),\label{eq:C_beta}
\end{align}
where the quantity $\beta(\cN)$ is the solution to the following SDP, with $N_{AB}$ denoting the Choi operator of $\cN$:
\begin{align}
\begin{aligned}
\text{\normalfont minimize: } & \tr(S_B)\\
\text{\normalfont subject to: } & -R_{AB} \leq N_{AB}^{T_B} \leq R_{AB}\\
& -\one_A \ox S_B \leq R_{AB}^{T_B} \leq \one_A \ox S_B.
\end{aligned}
\label{eq:beta-sdp}
\end{align}
The quantity $C_\beta(\cN)$ is a strong converse bound in the sense that for any sequence of classical codes with rate above the classical capacity $C(\cN)$, the success probability of that code sequence converges to 0 exponentially fast.

For the amplitude damping channel $\cA_p$, the quantity $C_\beta(\cA_p)$ constitutes the best known upper bound on $C(\cA_p)$ for $p\in [0,1/2]$ \cite{WXD16}, and it can be expressed in a closed form,
\begin{align}
C_\beta(\cA_p) = \log(1+\sqrt{1-p}).
\end{align}
We plot this bound in Figure~\ref {fig:amp-damp}, together with the Holevo information $\chi(\cA_p)$, which is a lower bound on the classical capacity $C(\cA_p)$.
Evidently, $C_\beta(\cA_p)$ is not tangent to the Holevo information $\chi(\cA_p)$. 
This indicates that a \emph{suitable} notion of approximate additivity of the Holevo information of $\cA_p$ would \emph{necessarily} improve upon $C_\beta(\cA_p)$ for low values of $p$ due to continuity.
Here, suitable means that the resulting upper bound is tangent to $\chi(\cA_p)$.
However, our arguments made above, together with the results from \cite{LLS17}, show that the approximation parameters defined in the present paper are \emph{not} suitable in this sense.
It therefore remains an interesting open question to find suitable approximation parameters for $\cA_p$ with the correct sublinear behavior in the damping parameter $p$.

\begin{figure}[t]
	\centering
	\includegraphics{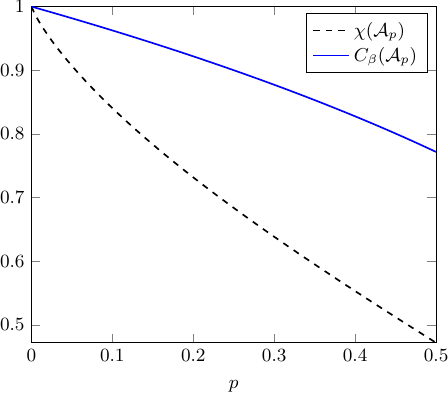}
	\caption{Lower and upper bound on the classical capacity $C(\cA_p)$ of the amplitude damping channel $\cA_p$ defined in \eqref{eq:amplitude-damping}.
		The Holevo information $\chi(\cA_p)$, a lower bound on $C(\cA_p)$, is plotted in black/dashed, and the strong converse upper bound $C_\beta(\cA_p)$, defined in \eqref{eq:C_beta} and equal to $\log(1+\sqrt{1-p})$, is plotted in blue/solid.}
	\label{fig:amp-damp}
\end{figure}

As mentioned above, the quantity $C_\beta(\cN)$ provides a strong bound on the classical capacity $C(\cN)$ of a quantum channel, as exhibited with the example of the amplitude damping channel.
Moreover, it can be easily computed for any channel $\cN$ via its SDP representation given in \eqref{eq:beta-sdp}.
For these reasons, we use $C_\beta(\cN)$ as a benchmark for upper bounds on the classical capacity of certain examples of channels in the following sections.
These bounds are obtained from Corollaries~\ref {cor:cov-upper-bounds},  \ref {cor:eb-upper-bound}, and \ref {cor:hadamard-upper-bound}, respectively.

\subsection{Convex mixture of amplitude damping and depolarizing noise}

\begin{figure*}[t]
	\centering
	\includegraphics{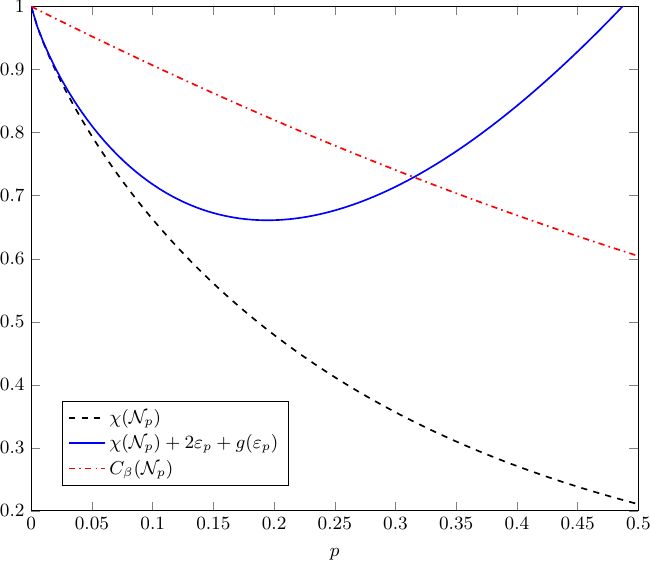}
	\caption{Upper and lower bounds on the classical capacity $C(\cN_p)$ of the channel $\cN_p$ defined in \eqref{eq:conv-mix-ampdamp-depol}. 
		The Holevo information $\chi(\cN_p)$, a lower bound on $C(\cN_p)$, is plotted in black/dashed, the upper bound $\chi(\cN_p) + 2\eps_p + g(\eps_p)$ on $C(\cN_p)$ from Corollary~\ref {cor:cov-upper-bounds}\MakeUppercase (\ref {item:cov-qubits}) with the covariance parameter $\eps_p\coloneqq \cov_\cP(\cN_p)=\frac{1}{2}p^2$ is plotted in blue/solid, and the strong converse upper bound $C_\beta(\cN_p)$ defined in \eqref{eq:C_beta} is plotted in red/dash-dotted.}
	\label{fig:conv-mix-ampdamp-depol}
\end{figure*}

We first consider a channel that is a convex mixture of an amplitude damping channel and a depolarizing channel.
More precisely, for $p\in [0,1]$ we define the channel
\begin{align}
\cN_p \coloneqq p \cA_p + (1-p) \cD_p,
\label{eq:conv-mix-ampdamp-depol}
\end{align}
where the amplitude damping channel $\cA_p$ is defined in \eqref{eq:amplitude-damping}, and $\cD_p$ denotes the qubit depolarizing channel 
\begin{align}
\cD_p(\rho) = (1-p) \rho + \frac{p}{3}(X\rho X + Y\rho Y + Z\rho Z).
\end{align}
The depolarizing channel $\cD_p$ is covariant under the full unitary group for all $p$, and hence also covariant with respect to the Pauli group $\cP$.
We therefore expect the channel $\cN_p$ to be almost covariant with respect to the Pauli group for small values of $p$, so that Corollary~\ref {cor:cov-upper-bounds} applies.
Indeed, we have the following immediate consequence of Proposition~\ref {prop:amp-damp-cov-param}:
\begin{corollary}
	The covariance parameter of the channel $\mathcal{N}_p$ defined in \eqref{eq:conv-mix-ampdamp-depol} with respect to the Pauli group $\cP$ is given by
	\begin{align}
	\operatorname{cov}_{\cP}(\mathcal{N}_p)=\frac{p^{2}}{2}.
	\end{align}
\end{corollary}
The resulting upper bound on $C(\cN_p)$ obtained via Corollary~\ref {cor:cov-upper-bounds}\MakeUppercase (\ref {item:cov-qubits}) is plotted in Figure~\ref {fig:conv-mix-ampdamp-depol} together with the Holevo information $\chi(\cN_p)$, and the strong converse upper bound $C_\beta(\cN_p)$ defined in \eqref{eq:C_beta}.
Numerical investigations (that were verified using the results from \cite{SSMR16}) showed that $\chi(\cN_p)$ is equal to the Holevo information $\chi(\cN_{p,\,\cP})$ of the twirled channel $\cN_{p,\,\cP}$ computed via Lemma~\ref {lem:1-design-holevo}, which enables us to use the stronger upper bound $\chi(\cN_p) + 2\eps_p + g(\eps_p)$ in Figure~\ref {fig:conv-mix-ampdamp-depol}.

\subsection{Composition of amplitude damping and dephasing noise}

\begin{figure*}[t]
	\centering
	\includegraphics{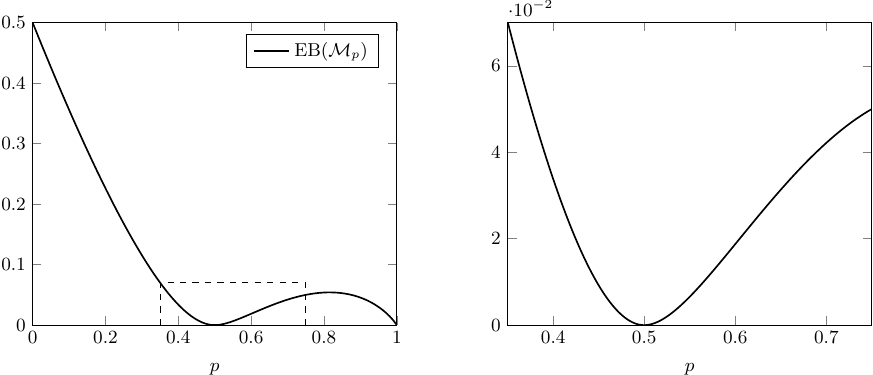}
	\caption{Plot of the entanglement-breaking parameter $\eb(\cM_p)$ of the channel $\cM_p$ defined in \eqref{eq:comp-ampdamp-dep} for the interval $p\in[0,1]$ (left) and zoomed in on the interval $p\in[0.35,0.75]$ (right).}
	\label{fig:comp-ampdamp-dep-eb-parameter}
\end{figure*}

We now consider the channel
\begin{align}
\cM_p \coloneqq \cA_p\circ \cZ_p,
\label{eq:comp-ampdamp-dep}
\end{align}
where $\cZ_p$ denotes the $Z$-dephasing channel
\begin{align}
\cZ_p(\rho) = (1-p)\rho + p Z\rho Z.
\end{align}
This channel was, e.g., considered by \textcite{ABDPST09} in the context of fault-tolerant quantum computation.

Let us explicitly calculate the Choi operator $M_{1/2}$ of $\cM_{1/2}$.
It is easy to see that $(\id_A\ox \cZ_{1/2})(\gamma) = |00\rangle\langle 00| + |11\rangle \langle 11|$, where $|00\rangle \equiv |0\rangle\ox|0\rangle$ and $|11\rangle \equiv |1\rangle \ox |1\rangle$.
The action of the Kraus operators $\lbrace K_1, K_2\rbrace$ of $\cA_{1/2}$ (given in \eqref{eq:amplitude-damping}) on the computational basis is
\begin{align}
K_1 |0\rangle &= |0\rangle & K_2 |0\rangle &= 0\\
K_1 |1\rangle &= \frac{1}{\sqrt{2}} |1\rangle & K_2|1\rangle &= \frac{1}{\sqrt{2}} |0\rangle,
\end{align}
from which we obtain $M_{1/2} = |00\rangle\langle 00| + |1\rangle\langle 1| \ox \pi_2$.
Since $M_{1/2}$ is manifestly separable, $\cM_{1/2}$ is entanglement-breaking, and in a neighborhood of $p=\frac{1}{2}$ we expect $\cM_p$ to be approximately entanglement-breaking according to Definition~\ref {def:approximate-entanglement-breaking}.
In Figure~\ref {fig:comp-ampdamp-dep-eb-parameter}, the entanglement-breaking parameter $\eb(\cM_p)$ is plotted as a function of $p$ for the whole interval $p\in[0,1]$ and zoomed in on the interval $p\in[0.35,0.75]$.
Evidently, the numerics suggest a quadratic dependence of $\eb(\cM_p)$ on $p$, i.e., $\eb(\cM_p)=O(p^2)$.
The resulting upper bound from Corollary~\ref {cor:eb-upper-bound} on the classical capacity $C(\cM_p)$ is plotted in Figure~\ref {fig:comp-ampdamp-dep}, both in terms of the original channel $\cM_p$ and the entanglement-breaking channel $\cM_p^{\mathrm{EB}}$ found by the SDP in \eqref{eq:EB-SDP} that is closest to $\cM_p$ in diamond distance. 
The plot also includes the Holevo information $\chi(\cM_p)$ and the strong converse upper bound $C_\beta(\cM_p)$.
The computation of $\chi(\cM_p)$ and $\chi(\cM_p^{\mathrm{EB}})$ was verified using the methods from \cite{SSMR16}.

\begin{figure*}[t]
	\centering
	\includegraphics{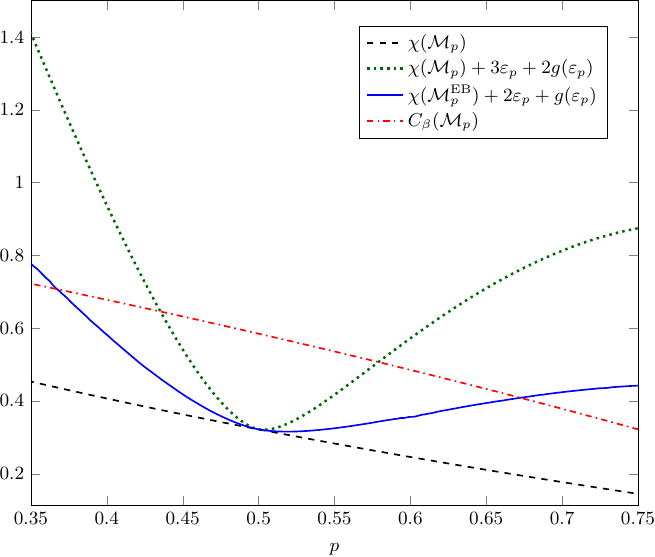}
	\caption{Upper and lower bounds on the classical capacity $C(\cM_p)$ of the channel $\cM_p$ defined in \eqref{eq:comp-ampdamp-dep}. 
		The Holevo information $\chi(\cM_p)$ is plotted in black/dashed, the upper bound $\chi(\cM_p) + 3\eps_p + 2g(\eps_p)$ on $C(\cM_p)$ from Corollary~\ref {cor:eb-upper-bound} with the entanglement-breaking parameter $\eps_p\coloneqq \eb(\cM_p)$ is plotted in green/dotted, the upper bound $\chi(\cM^{\mathrm{EB}}_p) + 2\eps_p + g(\eps_p)$ from Corollary~\ref {cor:eb-upper-bound} is plotted in blue/solid, and the strong converse upper bound $C_\beta(\cM_p)$ defined in \eqref{eq:C_beta} is plotted in red/dash-dotted.
		The channel $\cM^{\mathrm{EB}}_p$ is the entanglement-breaking channel closest in diamond distance to $\cM_p$ found by the SDP in \eqref{eq:EB-SDP}.}
	\label{fig:comp-ampdamp-dep}
\end{figure*}

\section{Concluding remarks and open problems}\label{sec:conclusion}
A quantum channel can be used in different contexts, and depending on the information-processing task there are different capacities characterizing the channel.
Most of these capacities --- including the classical, quantum, and private capacity --- are given in terms of regularized formulae that are intractable to evaluate for most channels.\footnote{A notable exception is the \emph{entanglement-assisted classical capacity} of a quantum channel, for which a single-letter formula was derived by \textcite{BSST99}.}
This is both good and bad news: Good, because entanglement between different uses of a quantum channel clearly enhances some its communication capabilities, giving rise to interesting effects such as superadditivity of coherent \cite{DSS98} and private information \cite{SRS08}, and superactivation of quantum capacity \cite{SY08};
bad, because the regularization severely limits our understanding of these capacities.
Hence, we need to resort to deriving useful lower and upper bounds on these capacities, as well as single out classes of channels for which the respective capacity formulae reduce to single-letter ones due to additivity of the underlying quantities (Holevo, coherent, and private information).

In the case of the quantum capacity, the largest class of channels with additive coherent information is the class of so-called informationally degradable channels \cite{CLS17}, which include degradable and conjugate degradable channels.\footnote{It is worthwhile to point out that there is no known example of an informationally degradable channel that is not also (conjugate) degradable. It remains an interesting open problem to determine whether the two sets are identical or not.}
To obtain upper bounds on the quantum capacity, \textcite{SSWR15} defined an approximate version of degradability, and showed that the desirable additivity properties of exactly degradable quantum channels are approximately preserved.
Their approach yields efficiently computable upper bounds on the quantum capacity, and works particularly well for so-called low-noise channels \cite{LLS17}, which are close in diamond norm to the ideal channel and hence `almost' degradable.

In this paper, we applied the approach of \cite{SSWR15} to the classical capacity.
In contrast to the quantum capacity, there are various distinct classes of channels with (weakly or strongly) additive Holevo information, including entanglement-breaking channels, unital qubit channels, depolarizing channels, and Hadamard channels.
This presented us with the option of defining multiple `approximation parameters': a covariance parameter (Definition~\ref {def:approximate-covariance}), an entanglement-breaking parameter (Definition~\ref {def:approximate-entanglement-breaking}), and two Hadamard parameters (Definition~\ref {def:approximate-hadamard} and Definition~\ref {def:approximate-hadamard-deg}).
For the first two parameters, we found interesting examples of channels (the channels $\cN_p$ and $\cM_p$ defined in \eqref{eq:conv-mix-ampdamp-depol} and \eqref{eq:comp-ampdamp-dep} of Section~\ref {sec:applications}, respectively) that exhibited the usefulness of our approach.
More generally, we note here that one can obtain single-letter upper bounds on the triple trade-off capacities \cite{HW10,HW10b,WH12,Wil13} of approximately Hadamard channels, given that Hadamard channels are quite special, as their triple-trade-off capacity regions are single-letter \cite{BHTW10,WH12,Wil13} in addition to their classical capacities.

We were not able to find an interesting channel for which the notion of approximate Hadamard-ness leads to good upper bounds on its classical capacity.
More precisely, we were looking for a channel $\cL_p$ defined in terms of a noise parameter $p\in[0,1]$, satisfying
\begin{align}
\Had(\cL_p)=O(p^a) \qquad \text{or} \qquad \Haddeg(\cL_p) =O(p^a)
\label{eq:good-hadamard-channel}
\end{align}
for some $a>1$, where $\Had(\cdot)$ and $\Haddeg(\cdot)$ are the Hadamard parameters defined in Definition~\ref {def:approximate-hadamard} and Definition~\ref {def:approximate-hadamard-deg}, respectively.
The results of \cite{LLS17} show that $a>1$ is necessary for the upper bound on $C(\cL_p)$ from Corollary~\ref {thm:approximate-additivity} to be non-trivial.
It remains an interesting open problem to find a channel satisfying \eqref{eq:good-hadamard-channel}.

On a more general note, we once again mention that our approach was not helpful for gaining a better understanding of the classical capacity of the amplitude damping channel $\cA_p$.
This is because for this channel all four approximation parameters introduced in this paper have at least linear behavior in the amplitude damping parameter $p$ (cf.~Section~\ref {sec:applications}), rendering the resulting upper bounds on $C(\cA_p)$ trivial.
We are curious as to whether a variation of the approximation methods of \cite{SSWR15} and the present paper might yield new insights in this matter, or whether a different approach, e.g., in the spirit of \cite{SS08,LDS17}, is needed.

\paragraph*{Acknowledgments.} We would like to thank David Sutter for providing MATLAB code for the results of \cite{SSMR16}, and Kun Wang for pointing out a minor issue with the numerical data for Figure~\ref{fig:ampdamp-hadamard}.
Part of this work was done during the workshop `Beyond I.I.D. in Information Theory’, hosted by the Institute for Mathematical Sciences, Singapore, 24-28 July 2017.
FL acknowledges support by the National Science Foundation under Grant Number 1125844, and appreciates the hospitality of the Hearne Institute for Theoretical Physics at Louisiana State University, Baton Rouge, where part of this work was done.
EK and MMW acknowledge support from the Office of Naval Research and the National Science Foundation.

\appendix
\section{Proof of \texorpdfstring{Lemma~\ref {lem:2-design-additivity}}{Lemma \ref{lem:2-design-additivity}}}\label{sec:bitwirl}
Before proving Lemma~\ref {lem:2-design-additivity}, we need the following result, whose proof we give for the sake of completeness:
\begin{lemma}[\cite{VW01}]\label{lem:bitwirl}
	Let $\cH$ and $\cH'$ be isomorphic Hilbert spaces with $d = \dim\cH = \dim\cH'$, and let $T\in\cB(\cH'\ox\cH)$ be a Hermitian operator.
	Then,
	\begin{multline}
	\int_\cU (U \ox \bar{U}) T (U\ox \bar{U})^\dagger d\mu(U) \\ = \frac{t-f}{d^2-1} \one_{\cH'\ox\cH} + \frac{d^2f-t}{d^2-1} \Phi,
	\end{multline}
	where $t\coloneqq \tr T$, $\Phi$ is a (normalized) maximally entangled state on $\cH'\ox\cH$, and $f\coloneqq \langle\Phi|T|\Phi\rangle$.
\end{lemma}

\begin{proof}
	We first consider an arbitrary Hermitian operator $S$ and its twirling by the product unitary $U\ox U$ (instead of $U\ox \bar{U}$).
	A standard argument using invariance of the Haar measure and Schur-Weyl duality shows that
	\begin{multline}\label{eq:schur-weyl}
	\int_{\cU} (U\ox U) S (U\ox U)^\dagger d\mu(U) \\ = \frac{ds-f'}{d^3-d} \one_{\cH'\ox\cH} + \frac{df'-s}{d^3-d} \mathbb{F},
	\end{multline}
	where $s=\tr S$ and $f'=\tr(S\mathbb{F})$, and the flip operator $\mathbb{F}$ is defined as the linear extension of its action on product vectors, $\mathbb{F}(|\phi\rangle\ox|\psi\rangle) \coloneqq |\psi\rangle \ox |\phi\rangle$.
	Let us denote the partial transpose of an operator $X\in\cB(\cH'\ox\cH)$ with respect to the second tensor factor by $X^{T_2}$.
	The following identities can easily be verified by inspection:
	\begin{align}
	\mathbb{F}^{T_2} &= d \Phi\\
	\tr S &= \tr S^{T_2} \\
	\tr\!\left(S\mathbb{F}\right)&= d\tr\left(S^{T_2} \Phi\right)\\
	\left((U\ox U) S (U\ox U)^\dagger\right)^{T_2} &= (U\ox \bar{U}) S^{T_2} (U\ox \bar{U})^\dagger
	\end{align}
	Taking the partial transpose on $\cH$ in \eqref{eq:schur-weyl} and using the above identities, we obtain
	\begin{align}
	&\int_{\cU} (U\ox \bar{U}) S^{T_2} (U\ox \bar{U})^\dagger d\mu(U)\nonumber\\
	&\qquad\qquad = \frac{ds-f'}{d^3-d} \one_{\cH'\ox\cH} + \frac{df'-s}{d^2-1} \Phi\\
	&\qquad\qquad = \frac{s-f}{d^2-1}\one_{\cH'\ox\cH} + \frac{d^2f-s}{d^2-1}\Phi,
	\end{align}
	where $f  = \tr\left(S^{T_2}\Phi\right) = f'/d$ and $s=\tr S=\tr S^{T_2}$.
	Choosing $S= T^{T_2}$ and noting that $s=t$ for this choice proves the claim.
\end{proof}

We are now ready to prove Lemma~\ref {lem:2-design-additivity}:
\begin{proof}[Proof of Lemma~\ref {lem:2-design-additivity}]
	The covariance of $\cN$ with respect to $\lbrace U_A(g)\rbrace_{g\in G}$ can be written as 
	\begin{align}
	\cN(\cdot) = U_A(g)^\dagger \cN(U_A(g) \cdot U_A(g)^\dagger) U_A(g).\label{eq:covariance-condition}
	\end{align}
	Let $\mathfrak{B}$ denote the basis for $\cH_A$ with respect to which the unnormalized maximally entangled vector $\gamma_{AA'}$ in \eqref{eq:gamma} is defined.
	Then, due to \eqref{eq:covariance-condition}, the Choi operator $N_{AB} \coloneqq (\id_{A}\ox \cN)(\gamma_{AA'})$ of $\cN$ is equal to 
	\begin{align}
	N_{AB} &= \frac{1}{|G|} \sum_{g\in G} (\bar{U}_A(g)\ox U_A(g)) N_{AB} (\bar{U}_A(g)\ox U_A(g))^\dagger \label{eq:choi-2-design}
	\end{align}
	where $\bar{X}$ denotes the complex conjugate of an operator $X$ with respect to the basis $\mathfrak{B}$.
	The formula \eqref{eq:choi-2-design} follows from the `transpose trick' identity $(A\ox \one)|\gamma\rangle = (\one \ox A^T) |\gamma\rangle$.
	Since $U_A(g)\in\cU(\cH_A)$ is a unitary two-design, we further have
	\begin{align}
	N_{AB} &= \frac{1}{|G|} \sum_{g\in G} (\bar{U}_A(g)\ox U_A(g)) N_{AB} (\bar{U}_A(g)\ox U_A(g))^\dagger\\
	&= \int_{\cU(\cH_A)} d\mu(U)\, (\bar{U}\ox U) N_{AB} (\bar{U}\ox U)^\dagger\\
	&= x \one_{AB} + y \Phi_{AB},\label{eq:depolarizing-Choi-state}
	\end{align}
	where $|\Phi\rangle_{AB}\coloneqq \frac{1}{\sqrt{|A|}} |\gamma\rangle_{AB}$, the last equality follows from Lemma~\ref {lem:bitwirl} in Appendix~\ref {sec:bitwirl}, and
	\begin{align}
	x &= \frac{|A|- f}{|A|^2-1} &
	y &= \frac{|A|^2f-|A|}{|A|^2-1}
	\end{align}
	with $f = \langle \Phi|N_{AB}|\Phi\rangle_{AB}$.
	Hence, setting $q = |A|x$ and noting that $\frac{y}{|A|} = 1-q$, we can rewrite \eqref{eq:depolarizing-Choi-state} as
	\begin{align}
	N_{AB} &= x \one_{AB} + y \Phi_{AB}\\
	&= |A|x \one_A \ox \pi_B + \frac{y}{|A|} \gamma_{AB}\\
	&= (\id_A \ox \cD_q)(\gamma_{AA'}),
	\end{align}
	where $\cD_q(\cdot) \coloneqq (1-q)\id + q\tr(\cdot)\pi$ is an $|A|$-dimensional depolarizing channel, which has strongly additive Holevo information as proved by \textcite{Kin03}.
\end{proof}

\section{Proof of \texorpdfstring{Theorem~\ref {thm:eps-hadamard}}{Theorem \ref{thm:eps-hadamard}}}\label{sec:eps-hadamard}
	The following proof of Theorem~\ref{thm:eps-hadamard} is closely related to the one of \cite[Theorem~20.4.1]{Wil13}  and previous methods employed in the proof of Theorem~\ref{prop:continuity}. 
	
	From the proof of
	\cite[Theorem~20.4.1]{Wil13}, we know that the Holevo
	information of a quantum channel $\mathcal{N}_{A\rightarrow B}$ can be written
	as%
	\begin{align}
	\chi(\cN) &= \max_{\{p_{X}(x),\psi^{x}\}} \left\lbrace H(B)_{\omega}-H(E|X)_{\omega} \right\rbrace\!,\label{eq:holevo-alt}%
	\intertext{where}
	\omega_{XBE} &=\sumi_{x}p_{X}(x)|x\rangle\langle x|_{X}\otimes\mathcal{V}%
	_{A\rightarrow BE}^{\mathcal{N}}(\psi_{A}^{x}),
	\end{align}
	$p_{X}$ is a probability distribution, $\{\psi_{A}^{x}\}_{x}$ is a set of pure
	states, and $\mathcal{V}_{A\rightarrow BE}^{\mathcal{N}}$ is a Stinespring dilation of $\mathcal{N}_{A\rightarrow B}$.
	
	It is convenient to distinguish between the output system $E$ of the
	complementary channel $\mathcal{N}_{c}$ and the output system $\tilde{E}$ of
	$\mathcal{D}_{B\rightarrow\tilde{E}}\circ\mathcal{N}_{A\rightarrow B}$, where
	$\tilde{E}\cong E$. For $n\in\mathbb{N}$, let%
	\begin{align}
	\rho_{XA^{n}}=\sumi_{x}p_{X}(x)|x\rangle\langle x|_{X}\otimes\psi_{A^{n}}^{x}%
	\end{align}
	be an arbitrary cq state, where each $\psi_{A^{n}}^{x}$ is a pure state. By
	assumption, we have that%
	\begin{align}
	\frac{1}{2}\Vert\mathcal{N}_{c}-\mathcal{D}_{B\rightarrow\widetilde{E}}%
	\circ\mathcal{N}_{A\rightarrow B}\Vert_{\diamond}\leq\varepsilon
	.\label{eq:diamond-norm-eps-hadamard}%
	\end{align}
	We define the following cq states for $1\leq t\leq n-1$:
	\begin{align}
	\sigma_{XB^{n}E^{n}}^{0} &\coloneqq\left(  \operatorname{id}_{X}\otimes\left[
	\mathcal{V}^{\mathcal{N}}\right]  ^{\otimes n}\right)
	(\rho_{XA^{n}}),\\
	\sigma_{X\tilde{E}_{\leq t}E_{>t}}^{t} &\coloneqq \left(  \operatorname{id}_{X}\otimes\left[  \mathcal{D}\circ\mathcal{N}\right]  ^{\otimes t}\otimes
	\mathcal{N}_{c}^{\otimes(n-t)}\right)  (\rho_{XA^{n}}),\\
	\sigma_{XF^{n}\tilde{E}^{n}}^{n} &\coloneqq\left(  \operatorname{id}_{X}\otimes\left[
	\mathcal{U}^{\mathcal{D}}\circ\mathcal{N}\right]  ^{\otimes n}\right)  (\rho_{XA^{n}}),
	\end{align}
	where $\cV^\cN \equiv \cV^\cN_{A\to BE}$ and $\cU^\cD\equiv \cU^\cD_{B\rightarrow F\tilde{E}}$ is an isometric extension of $\cD_{B\rightarrow \tilde{E}}$.
	
	The multi-letter version of \eqref{eq:holevo-alt} that we are interested in
	bounding is%
	\[
	H(B^{n})_{\sigma^{0}}-H(E^{n}|X)_{\sigma^{0}}=H(F^{n}\tilde{E}^{n}%
	)_{\sigma^{n}}-H(E^{n}|X)_{\sigma^{0}},
	\]
	where the equality holds because entropy is invariant with respect to applying the isometry $\mathcal{U}^{\mathcal{D}}$.
	For an entanglement-breaking channel $\mathcal{D}_{B\rightarrow\tilde{E}}%
	\circ\mathcal{N}_{A\rightarrow B}$ and for all integers $n>1$, there exist
	states $\omega_{X_{i}F_{i}\tilde{E}_{i}}^{i}$ for $1\leq i\leq n$, each
	having the form%
	\begin{align}
	\omega_{X_{i}F_{i}\tilde{E}_{i}}^{i}\coloneqq \sum_{x_{i}}p_{X_{i}}(x_{i}%
	)|x_{i}\rangle\langle x_{i}|_{X_{i}}\otimes(\cU^\cD_i\circ\mathcal{N}_i)(\phi_{A_{i}}^{x_{i}})
	\end{align}
	with $\cU^\cD_i\equiv \mathcal{U}_{B_{i}\rightarrow F_{i}\tilde{E}_{i}}^{\mathcal{D}}$ and $\cN_i\equiv\mathcal{N}_{A_{i}\rightarrow B_{i}}$, 
	such that the following entropy inequalities hold \cite[Theorem~20.4.1]{Wil13}:
	\begin{align}
	H(F^{n}\tilde{E}^{n})_{\sigma^{n}}  & \leq\sum_{i=1}^{n}H(F_{i}\tilde{E}%
	_{i})_{\omega^{i}} \\
	&= H(F\tilde{E}|Z)_{\omega},\label{eq:eps-had-subadd}\\
	-H(\tilde{E}^{n}|X)_{\sigma^{n}}  & \leq-\sum_{i=1}^{n}H(\tilde{E}_{i}%
	|X_{i})_{\omega^{i}}\\
	&=-nH(\tilde{E}|\tilde{X}Z)_{\omega}%
	,\label{eq:eps-had-subadd-env}%
	\end{align}
	where the state $\omega_{Z\tilde{X}F\tilde{E}}$ is defined as%
	\begin{align}
	\omega_{Z\tilde{X}F\tilde{E}}\coloneqq\frac{1}{n}\sum_{i=1}^{n}|i\rangle\langle
	i|_{Z}\otimes\omega_{X_{i}F_{i}\tilde{E}_{i}}^{i},
	\end{align}
	and the register $\tilde{X}$ is taken to be large enough to contain the
	contents of the largest $X_{i}$. Applying
	\eqref{eq:diamond-norm-eps-hadamard}, we find that
	\begin{align}
	\frac{1}{2}\Vert\sigma_{X\tilde{E}_{\leq t-1}E_{>t-1}}^{t-1}-\sigma
	_{X\tilde{E}_{\leq t}E_{>t}}^{t}\Vert_{1}\leq\varepsilon
	,\label{eq:eve-continuity-states-eps-had}%
	\end{align}
	for $1\leq t\leq n$ by the definition of the diamond norm. Moreover,
	$\sigma_{X\tilde{E}_{\leq t-1}E_{>t-1}}^{t-1}$ and $\sigma_{X\tilde{E}_{\leq
			t}E_{>t}}^{t}$ have the same marginals on the systems $X\tilde{E}_{<t}E_{>t}$.
	We get that%
	\begin{multline}
	\left\vert H(E_{t}|\tilde{E}_{<t}E_{>t}X)_{\sigma^{t-1}}-H(\tilde{E}%
	_{t}|\tilde{E}_{<t}E_{>t}X)_{\sigma^{t}}\right\vert \\
	\leq2\varepsilon
	\log|E|+g(\varepsilon),
	\end{multline}
	by applying \eqref{eq:eve-continuity-states-eps-had} and \cite[Lemma~2]%
	{Win15}. By reasoning similar to that in the proof of \cite[Theorem~3.4]%
	{SSWR15}, consider that%
	\begin{align}
	&  -H(E^{n}|X)_{\sigma^{0}}\nonumber\\
	&  =-H(\tilde{E}^{n}|X)_{\sigma^{n}}-H(E^{n}|X)_{\sigma^{0}}+H(\tilde{E}%
	^{n}|X)_{\sigma^{n}}\nonumber\\
	&  =-H(\tilde{E}^{n}|X)_{\sigma^{n}}\nonumber\\
	&  \quad-\sum_{t=1}^{n}\left[  H(E_{t}|\tilde{E}_{<t}E_{>t}X)_{\sigma^{t-1}%
	}-H(\tilde{E}_{t}|\tilde{E}_{<t}E_{>t}X)_{\sigma^{t}}\right]  \nonumber\\
	&  \leq-H(\tilde{E}^{n}|X)_{\sigma^{n}}+n\left(  2\varepsilon\log
	|E|+g(\varepsilon)\right)  \nonumber\\
	&  \leq -nH(\tilde{E}|\tilde{X}Z)_{\omega}+n\left(  2\varepsilon\log
	|E|+g(\varepsilon)\right)  .\label{eq:eps-had-env-upper}%
	\end{align}
	The final inequality is a consequence of \eqref{eq:eps-had-subadd-env}. We can
	now bound the quantity $H(F^{n}\tilde{E}^{n})_{\sigma^{n}}-H(E^{n}%
	|X)_{\sigma^{0}}$ of interest from above as follows:%
	\begin{align}
	& H(F^{n}\tilde{E}^{n})_{\sigma^{n}}-H(E^{n}|X)_{\sigma^{0}}\\
	& \leq n\left[  H(F\tilde{E}|Z)_{\omega}-H(\tilde{E}|\tilde{X}Z)_{\omega
	}\right]  +n\left(  2\varepsilon\log|E|+g(\varepsilon)\right)  \\
	& \leq n\left[  H(F\tilde{E})_{\omega}-H(\tilde{E}|\tilde{X}Z)_{\omega
	}\right]  +n\left(  2\varepsilon\log|E|+g(\varepsilon)\right)  \\
	& \leq n \max_{\{p_{X}(x),\psi^{x}\}} \left\lbrace H(F\tilde{E})_{\xi}-H(\tilde
	{E}|X)_{\xi} \right\rbrace  \\ & \qquad +n\left(  2\varepsilon\log|E|+g(\varepsilon)\right)\!.
	\end{align}
	The first inequality follows from \eqref{eq:eps-had-subadd} and
	\eqref{eq:eps-had-env-upper}. The second inequality follows from the chain
	rule. The final inequality follows because the state $\omega_{Z\tilde
		{X}F\tilde{E}}$ is a particular one to consider for optimizing the entropy
	difference. Furthermore, it suffices to optimize over pure states due to
	another application of the data processing inequality. Since we have shown the above chain of inequalities for an arbitrary positive integer $n$ and an arbitrary initial ensemble at the channel input, we can conclude the statement of the theorem. \qedhere

\section{Properties of the amplitude damping channel}\label{sec:ampdamp-properties}
We first prove Proposition~\ref {prop:amp-damp-cov-param}, which states that $\cov_{\cP}(\mathcal{A}_{p})=\frac{p}{2}.$
\begin{proof}[Proof of Proposition~\ref {prop:amp-damp-cov-param}]
	We set $G=\cP$ and argue by employing Lemma~\ref{lemma:cov-critical-step}, proving the
	equality in \eqref{eq:amp-damping-cov} by first showing that the maximally
	entangled qubit state achieves the diamond norm and then that the resulting
	trace distance is equal to $\frac{p}{2}$, i.e.,%
	\begin{align}
	\frac{1}{2}\left\Vert \mathcal{A}_{p}-\mathcal{A}_{p,G}\right\Vert _{\diamond
	}=\frac{1}{2}\left\Vert \operatorname{id}_{R}\otimes\left[  \mathcal{A}%
	_{p}-\mathcal{A}_{p,G}\right]  (\Phi_{RA})\right\Vert _{1}=\frac{p}{2}.
	\end{align}
		
	First, consider a general channel $\mathcal{N}$ and its twirled version
	$\mathcal{N}_{G}$ with respect to some covariance group $\{\left(  U_{g}%
	,V_{g}\right)  \}_{g}$, where $\left\{  U_{g}\right\}  _{g}$ is a one-design.
	We again use the notation $\cN^g \coloneqq \mathcal{V}^{g\dag}\circ\mathcal{N}\circ\mathcal{U}^{g}$.
	Taking the generalized divergence to be the trace distance, $\psi_{RA}$ an
	arbitrary pure state, and applying Lemma~\ref{lemma:cov-critical-step}, we
	find that%
	\begin{align}
	&  \left\Vert \left(  \operatorname{id}_{R}\otimes\mathcal{N}\right)
	(\Phi_{RA})-\left(  \operatorname{id}_{R}\otimes\mathcal{N}_{G}\right)
	(\Phi_{RA})\right\Vert _{1}\nonumber\\
	& \quad \geq\left\Vert \frac{1}{\left\vert G\right\vert }\sumi_{g} \cN^g \left(  \psi_{RA}\right)  \otimes
	|g\rangle\langle g|- \mathcal{N}%
	_{G}   (\psi_{RA})\otimes\pi_{G}\right\Vert _{1}\\
	& \quad =\frac{1}{\left\vert G\right\vert }\sum_{g}\left\Vert  \cN^g   \left(  \psi_{RA}\right)  - \mathcal{N}_{G} (\psi_{RA})\right\Vert
	_{1},
	\end{align}
	where $\pi_{G}$ denotes the maximally mixed state (note that for the second
	channel, we get the maximally mixed state here because the channel is already
	symmetrized with respect to the covariance group). The second equality follows
	from how the trace norm decomposes when acting on block-diagonal operators.
	
	Applying the above inequality to the amplitude damping channel $\mathcal{A}%
	_{p}$ and the symmetrized channel $\mathcal{A}_{p,G}$ with covariance group
	given by the Pauli group on the input and output, and noting that
	$\mathcal{A}_{p}$ is covariant with respect to $\sigma_{Z}$, we find that%
	
	\begin{widetext}
		\begin{align}
		\left\Vert \mathcal{A}_{p}(\Phi_{RA})-\mathcal{A}_{p,G}(\Phi_{RA})\right\Vert
		_{1} \geq \frac{1}{2}\left\Vert \mathcal{A}_{p}(\psi_{RA})-\mathcal{A}%
		_{p,G}(\psi_{RA})\right\Vert _{1} +\frac{1}{2}\left\Vert \sigma_{X}\mathcal{A}_{p}(\sigma_{X}\psi_{RA}\sigma
		_{X})\sigma_{X}-\mathcal{A}_{p,G}(\psi_{RA})\right\Vert _{1}.
		\end{align}
		For an amplitude damping channel, consider that $\mathcal{A}_{p,G}%
		(\cdot)=\frac{1}{2}\mathcal{A}_{p}(\cdot)+\frac{1}{2}\sigma_{X}\mathcal{A}%
		_{p}(\sigma_{X}\left(  \cdot\right)  \sigma_{X})\sigma_{X}$. 
		This implies that%
	\begin{align}
	\left\Vert \sigma_{X}\mathcal{A}_{p}(\sigma_{X}\psi_{RA}\sigma_{X})\sigma_{X}-\mathcal{A}_{p,G}(\psi_{RA})\right\Vert _{1} &=\bigg\Vert \sigma_{X}\mathcal{A}_{p}(\sigma_{X}\psi_{RA}\sigma_{X}%
	)\sigma_{X} -\frac{1}{2}\left[  \mathcal{A}_{p}(\psi_{RA})+\sigma
	_{X}\mathcal{A}_{p}(\sigma_{X}\psi_{RA}\sigma_{X})\sigma_{X}\right]
	\bigg\Vert _{1}\\
	&  =\frac{1}{2}\left\Vert \mathcal{A}_{p}(\psi_{RA})-\sigma_{X}\mathcal{A}%
	_{p}(\sigma_{X}\psi_{RA}\sigma_{X})\sigma_{X}\right\Vert _{1}\\
	&  =\left\Vert \mathcal{A}_{p}(\psi_{RA})-\frac{1}{2}\left[  \mathcal{A}%
	_{p}(\psi_{RA})+\sigma_{X}\mathcal{A}_{p}(\sigma_{X}\psi_{RA}\sigma_{X}%
	)\sigma_{X}\right]  \right\Vert _{1}\\
	&  =\left\Vert \mathcal{A}_{p}(\psi_{RA})-\mathcal{A}_{p,G}(\psi
	_{RA})\right\Vert _{1}.
	\end{align}
	\end{widetext}
	Combining with the above, this implies that the following inequality holds for
	any state $\psi_{RA}$:%
	\begin{multline}
	\left\Vert \mathcal{A}_{p}(\Phi_{RA})-\mathcal{A}_{p,G}(\Phi_{RA})\right\Vert
	_{1} \\ 
	\geq\left\Vert \mathcal{A}_{p}(\psi_{RA})-\mathcal{A}_{p,G}(\psi
	_{RA})\right\Vert _{1}.
	\end{multline}
	This establishes that the diamond norm is achieved by the maximally entangled
	state. We can then use the above again to see that%
	\begin{multline}
	\left\Vert \mathcal{A}_{p}(\Phi_{RA})-\mathcal{A}_{p,G}(\Phi_{RA})\right\Vert
	_{1} \\
	=\frac{1}{2}\left\Vert \mathcal{A}_{p}(\Phi_{RA})-\sigma_{X}%
	\mathcal{A}_{p}(\sigma_{X}\Phi_{RA}\sigma_{X})\sigma_{X}\right\Vert _{1}.
	\end{multline}
	To compute the value of the right-hand side, recall that the Kraus operators
	of the amplitude damping channel $\mathcal{A}_{p}$ are given by%
	\begin{align}
	K_1 &= |0\rangle\langle0|+\sqrt{1-p}|1\rangle\langle 1| \\ 
	K_2 &= \sqrt{p}|0\rangle
	\langle 1|,
	\end{align}
	which implies that the Kraus operators of $\sigma_{X}\mathcal{A}_{p}%
	(\sigma_{X}\left(  \cdot\right)  \sigma_{X})\sigma_{X}$ are given by%
	\begin{align}
	L_1 &= |1\rangle\langle1|+\sqrt{1-p}|0\rangle\langle 0|\\
	L_2 &= \sqrt{p}|1\rangle \langle 0|.
	\end{align}
	Applying these Kraus operators to the maximally entangled state $|\Phi\rangle = \frac{1}{\sqrt{2}}(|00\rangle + |11\rangle)$ leads
	to%
	\begin{align}
	K_1 \Phi K_1^\dagger &=\frac{1}{2} \big(  |00\rangle\langle00|+\sqrt{1-p}|00\rangle
	\langle 11|\nonumber\\
	&\qquad {}+\sqrt{1-p}|11\rangle\langle00|+\left(  1-p\right)  |11\rangle
	\langle 11| \big) \\
	K_2 \Phi K_2^\dagger &=\frac{p}{2}|10\rangle\langle10|\\
	L_1\Phi L_1^\dagger &= \frac{1}{2}\big(  \left(  1-p\right)  |00\rangle\langle00|+\sqrt
	{1-p}|00\rangle\langle 11| \nonumber\\
	&\qquad{}+\sqrt{1-p}|11\rangle\langle 00|+|11\rangle
	\langle11|\big)\\
	L_2\Phi L_2^\dagger &= \frac{p}{2}|01\rangle\langle01|
	\end{align}
	Then we find that%
	\begin{multline}
	\mathcal{A}_{p}(\Phi_{RA})-\sigma_{X}\mathcal{A}_{p}(\sigma_{X}\Phi
	_{RA}\sigma_{X})\sigma_{X}\\
	=\frac{p}{2}\left(  |00\rangle\langle00|+|10\rangle\langle10|+|01\rangle
	\langle01|-|11\rangle\langle11|\right)  ,
	\end{multline}
	implying that%
	\begin{align}
	\left\Vert \mathcal{A}_{p}(\Phi_{RA})-\sigma_{X}\mathcal{A}_{p}(\sigma_{X}%
	\Phi_{RA}\sigma_{X})\sigma_{X}\right\Vert _{1}=2p
	\end{align}
	and in turn that%
	\begin{align}
	\left\Vert \mathcal{A}_{p}(\Phi_{RA})-\mathcal{A}_{p,G}(\Phi_{RA})\right\Vert
	_{1}=p.
	\end{align}
	This concludes the proof.
\end{proof}

Next, we prove Proposition~\ref {prop:amp-damp-eb-param}, which states that 
\begin{align}
\eb(\cA_p) \geq f(p)\coloneqq \frac{(1-p)(2\sqrt{1-p}-p)}{4(1-p)-p^2} \geq \frac{1-p}{2}.
\end{align}
Note that numerics suggest that $f(p)$ is in fact optimal, i.e., $\eb(\cA_p) = f(p)$.

\begin{proof}[Proof of Proposition~\ref {prop:amp-damp-eb-param}]
	We prove the lower bound on $\eb(\cA_p)$ by judiciously choosing a feasible point of the primal SDP for $\eb(\cA_p)$.
	The primal SDP can be obtained from the dual formulation in \eqref{eq:EB-SDP} by standard techniques (we use the notation $\langle X_A,Y_A\rangle \coloneqq \tr(X_A^\dagger Y_A)$ for the Hilbert-Schmidt inner product):
	\begin{align}
	\begin{aligned}
	\text{\normalfont maximize: } & \frac{1}{2}\left( \langle N_{AB},A_p\rangle - \tr(H_A) \right) \\
	\text{\normalfont subject to: } & \tr(M_A) \leq 2\\
	& N_{AB} \leq M_A \ox \one_B\\
	& N_{AB} + P_{AB}^{T_B} \leq H_A \ox \one_B\\
	& M_A \geq 0\\
	& N_{AB} \geq 0\\
	& P_{AB} \geq 0\\
	& H_A = H_A^\dagger,	
	\end{aligned}\label{eq:EB-SDP-primal}
	\end{align}
	where $A_p$ is the Choi operator of the amplitude damping channel $\cA_p$,
	\begin{align}
	A_p = (\id\ox\cA_p)(\gamma) = \begin{pmatrix}
	1 & 0 & 0 & \sqrt{1-p} \\ 0 & 0 & 0 & 0\\ 0 & 0 & p & 0\\ \sqrt{1-p} & 0 & 0 & 1-p
	\end{pmatrix}\!.
	\label{eq:ampdamp-choi}
	\end{align}
	We now make the following ansatz for the variables appearing in the primal SDP in \eqref{eq:EB-SDP-primal} (with a dot representing a zero for improved readability):\footnote{The ansatz in \eqref{eq:eb-primal-sdp-ansatz} and \eqref{eq:eb-primal-sdp-var-constr}, which seems rather ad-hoc at first, was determined by carefully analyzing the optimal numerical solutions to the SDP in \eqref{eq:EB-SDP-primal}.}
	\begin{align}
	N_{AB} &= \begin{pmatrix}
	q_1 & . & . & r\\
	. & . & . & .\\
	. & . & . & .\\
	r & . & . & q_2
	\end{pmatrix} & 
	P_{AB} &= \begin{pmatrix}
	. & . & . & .\\ . & q_1 & -r & . \\ . & -r & q_2 & . \\ . & . & . & .
	\end{pmatrix} \\
	H_A &= \begin{pmatrix}
	q_1 & .\\ .& q_2
	\end{pmatrix} & 
	M_A&= \begin{pmatrix}
	q_1 + r & . \\ . & q_2 + r
	\end{pmatrix}\!,
	\label{eq:eb-primal-sdp-ansatz}
	\end{align}
	where $q_1,q_2,r\geq 0$ are real parameters satisfying
	\begin{align}
	r^2 = q_1 q_2 \qquad \text{and} \qquad r = 1-\frac{1}{2}(q_1 + q_2).
	\label{eq:eb-primal-sdp-var-constr}
	\end{align}
	It is easy to check that, imposing the conditions \eqref{eq:eb-primal-sdp-var-constr} on $q_1,q_2,r$, the choices of $N_{AB},P_{AB},H_A,M_A$ in \eqref{eq:eb-primal-sdp-ansatz} satisfy all constraints in \eqref{eq:EB-SDP-primal}.
	Moreover, for these choices the objective function reduces to
	\begin{align}
	\frac{1}{2}\left( \langle N_{AB},A_p\rangle - \tr(H_A) \right) = \frac{1}{2}\left(2r\sqrt{1-p} - pq_2\right)\!.\label{eq:objective}
	\end{align}
	With the constraints \eqref{eq:eb-primal-sdp-var-constr}, the objective function \eqref{eq:objective} can be determined as a function $f(p)$ using the method of Lagrange multipliers, which we carried out in the attached Mathematica notebook \verb+amp_damp_EB_parameter.nb+.
	This yields
	\begin{align}
	f(p) = \frac{(1-p)(2\sqrt{1-p}-p)}{4(1-p)-p^2}.
	\end{align}
	It is furthermore easy to see that $f(p)\geq \frac{1}{2}(1-p)$, which proves the claim.
\end{proof}

To conclude, we prove the fact that the amplitude damping channel $\cA_p$ is a \emph{low-noise channel} in the sense of \cite{LLS17}, i.e., $\|\id - \cA_p\|_\diamond = 2p$.

\begin{proof}[Proof of Proposition~\ref {prop:ampdamp-low-noise}]
	By sending in the state $|1\rangle
	\langle1|$ (with trivial purification $|11\rangle\langle 11|$) to both channels, we first obtain the following lower bound:%
	\begin{align}
	\left\Vert \operatorname{id}-\mathcal{A}_{p}\right\Vert _{\diamond}\geq 2p.
	\end{align}
	Indeed, consider that%
	\begin{align}
	\mathcal{A}_{p}(|1\rangle\langle1|)=(1-p)|1\rangle\langle1|+p|0\rangle
	\langle0|,
	\end{align}
	so that%
	\begin{align}
	\operatorname{id}(|1\rangle\langle1|)-\mathcal{A}_{p}(|1\rangle\langle1|)  &
	=|1\rangle\langle1|-(1-p)|1\rangle\langle1|-p|0\rangle\langle0|\\
	&  =p|1\rangle\langle1|-p|0\rangle\langle0|,
	\end{align}
	and then 
	\begin{align}
	&\|(\id\ox\id)(|11\rangle\langle 11|) - (\id\ox \cA_p)(|11\rangle\langle 11|)\|_1 \nonumber \\
	&\qquad\qquad  = \left\Vert \operatorname{id}(|1\rangle\langle1|)-\mathcal{A}%
	_{p}(|1\rangle\langle1|)\right\Vert _{1}\\
	&\qquad\qquad =2p.
	\end{align}
	
	We thus need to prove that the value of $2p$ is indeed optimal.
	One way to show this is by SDP duality as follows.
	Consider the bipartite operator
	\begin{align}
	Z_{AB} \coloneqq \begin{pmatrix}
	\frac{q^2}{p} & 0 & 0 & q \\ 0 & p- \frac{q^2}{p} & 0 & 0\\ 0 & 0 & 0 & 0 \\ q & 0 & 0 & p
	\end{pmatrix}\!,
	\end{align}
	where $q\coloneqq 1-\sqrt{1-p}$.
	It is easy to check (see \verb+amp_damp_feasible_Z.nb+ in the ancillary files) that
	\begin{align}
	Z_{AB} \geq 0 \qquad \text{and}\qquad Z_{AB} \geq \gamma - A_p,
	\label{eq:Z-feasibility}
	\end{align}
	where $\gamma\equiv|\gamma\rangle\langle\gamma|$ is the Choi operator of $\id$, and where $A_p$ is the Choi operator of $\cA_p$ as given in \eqref{eq:ampdamp-choi}.
	The two conditions \eqref{eq:Z-feasibility} show that $Z_{AB}$ is a feasible point for the dual SDP of $\frac{1}{2}\|\id - \cA_p\|_\diamond$ in \eqref{eq:diamondnorm-SDP}.
	Furthermore, $\tr_BZ_{AB} = p \one_A$, which shows that $\frac{1}{2}\|\id - \cA_p\|_\diamond \leq p$, concluding the proof.
	
	In the following, we also give an analytical proof of the fact that $\|\id - \cA_p\|_\diamond \leq 2p$.
	We consider only the interval $p \in(0,1]$ because otherwise the diamond norm is trivially equal to zero, given that the amplitude damping channel $\mathcal{A}%
	_{p}$ becomes the identity channel for $p=0$. The identity channel and the
	amplitude damping channel are jointly covariant with respect to the group
	$\{I,\sigma_{Z}\}$. Applying Proposition~\ref{prop:covariance}, we find that
	the optimal state for the diamond norm of the difference $\operatorname{id}%
	-\mathcal{A}_{p}$ takes the form:%
	\begin{align}
	|\phi^{q}\rangle\equiv\sqrt{q}|0\rangle|0\rangle+\sqrt{1-q}|1\rangle|1\rangle,
	\end{align}
	for some $q\in\left[  0,1\right]  $. By following steps similar to those in
	the proof of Proposition~\ref{prop:amp-damp-cov-param}, we find that%
	\begin{multline}
	|\phi^{q}\rangle\langle\phi^{q}|-\left(  \operatorname{id}\otimes
	\mathcal{A}_{p}\right)  \left(  |\phi^{q}\rangle\langle\phi^{q}|\right) \\
	\qquad {} =\sqrt{q\left(  1-q\right)  }\left(  1-\sqrt{1-p}\right)  \left(
	|00\rangle\langle11|+|11\rangle\langle00|\right) \\
	+p\left(  1-q\right)  \left(  |11\rangle\langle11|-|10\rangle\langle
	10|\right).
	\end{multline}
	One can compute that the eigenvalues of the above matrix are given by%
	\begin{align}
	\left\lbrace 0, p(q-1), \frac{1}{2}\left(  p\left(  1-q\right)  \pm\sqrt{\left(  1-q\right)r(p,q) 	}\right)\right\rbrace,
	\end{align}
	where 
	\begin{align}
	r(p,q) \coloneqq p^{2}\left(  1-q\right)  +8q\left(  1-\sqrt{1-p}\right)  -4pq.
	\end{align}
	This implies that the trace norm $\||\phi^{q}\rangle\langle\phi^{q}|-\left(  \operatorname{id}\otimes
	\mathcal{A}_{p}\right)  \left(  |\phi^{q}\rangle\langle\phi^{q}|\right)\|_1$ is equal to%
	\begin{align}
	f(p,q) \coloneqq p\left(  1-q\right)  +\sqrt{\left(  1-q\right)  r(p,q)}.
	\end{align}
	We note that $r(p,q)\geq0$ for all $p\in(0,1]$ and $q\in\left[  0,1\right]  $, which follows
	because $r(p,0)=p^{2}\geq0$ and%
	\begin{align}
	\frac{d}{dq}r(p,q)=8(1-\sqrt{1-p})-p\left(  4+p\right)  \geq
	0,\label{eq:particular-p-q-ineq}%
	\end{align}
	on $p\in(0,1]$, implying that $r(p,q)$ is monotone increasing in $q$ on
	the interval $\left[  0,1\right]  $, for all fixed $p\in\left(  0,1\right]  $.
	Consider that%
	\begin{align}
	f(p,0)=2p.
	\end{align}
	Our aim is to show that for fixed $p\in\left(  0,1\right]  $, the function
	$f(p,q)$ is monotone decreasing in $q$ on the interval $\left[  0,1\right]  $.
	If this holds, then we can conclude the statement of the proposition. To this
	end, we will compute $\frac{d}{dq}f(p,q)$ and show that $\frac{d}%
	{dq}f(p,q)\leq0$ for all $p\in(0,1]$ and $q\in\left[  0,1\right]  $. Consider that%
	\begin{align}
	&\frac{d}{dq}f(p,q) \nonumber \\
	&\quad =-p-\frac{r(p,q)+\left(  1-q\right)  \left(  4p+p^{2}%
		-8\left[  1-\sqrt{1-p}\right]  \right)  }{2\sqrt{\left(  1-q\right)  r(p,q)}}.
	\end{align}
	We have
	\begin{align}
	&\left.  \frac{d}{dq}f(p,q)\right\vert _{q=0} \nonumber\\ 
	&\quad =\frac{2}{p}\left(  2\left[
	1-\sqrt{1-p}\right]  -p\left(  1+p\right)  \right)  \leq
	0\label{eq:first-deriv-q-0}%
	\end{align}
	for all $p\in(0,1]$. We then compute%
	\begin{align}
	\frac{d^{2}}{dq^{2}}f(p,q) &= -\frac{4s(p)  \sqrt{\left(  1-q\right)  r(p,q)}%
	}{\left(  1-q\right)  ^{2}\left[  r(p,q)\right]  ^{2}},
	\end{align}
	where $s(p)\coloneqq 8\left(  1-\sqrt{1-p}\right)  +p\left(  p+4\sqrt{1-p}-8\right)$.
	Since $s(p) \geq0$ for $p\in\left(  0,1\right]  $, we conclude that $\frac{d^{2}}{dq^{2}%
	}f(p,q)\leq0$. This implies that for fixed $p\in\left(  0,1\right]  $,
	$\frac{d}{dq}f(p,q)$ is monotone decreasing in $q\in\left[  0,1\right]  $.
	This in turn implies, by combining with \eqref{eq:first-deriv-q-0}, that
	$\frac{d}{dq}f(p,q)\leq0$. This concludes the proof.
\end{proof}

\bibliographystyle{apsrev}
\bibliography{references}

\end{document}